\documentclass[11pt]{article}

\usepackage[utf8]{inputenc}	
\usepackage[T1]{fontenc}
\usepackage{lmodern}
\usepackage{amsmath}
\usepackage[english]{babel}
\usepackage{amssymb}
\usepackage{amsthm}
\usepackage{mathtools}
\usepackage{esdiff}
\usepackage[normalem]{ulem}
\usepackage{cancel}
\usepackage{microtype}
\usepackage{paralist}
\usepackage{skak}
\usepackage{bbm}
\usepackage[paper=a4paper,left=25mm,right=25mm,top=25mm,bottom=25mm]{geometry}
\usepackage{stmaryrd}
\usepackage{hyperref}
\usepackage[dvipsnames]{xcolor} 
\usepackage{tikz-cd}
\usepackage{MnSymbol}
\usepackage{fancyhdr}

\usepackage{thmtools}
\declaretheoremstyle[headfont=\normalfont\bfseries]{bfthmstyle}

\declaretheorem[sharenumber=Theorem,style=bfthmstyle]{Lemma}
\declaretheorem[sharenumber=Theorem,style=bfthmstyle]{Proposition}
\declaretheorem[sharenumber=Theorem,style=bfthmstyle]{Corollary}

\declaretheoremstyle[headfont=\normalfont\bfseries,qed={$\diamond$}]{bfdefstyle}
\declaretheorem[sharenumber=Theorem,style=bfdefstyle]{Definition}
\declaretheorem[sharenumber=Theorem,style=bfdefstyle]{Definition-Proposition}
\declaretheorem[sharenumber=Theorem,style=bfdefstyle]{Definition-Lemma}
\declaretheorem[sharenumber=Theorem,style=bfdefstyle]{Example}

\declaretheoremstyle[headfont=\normalfont\bfseries,qed={$\diamond$}]{bfremstyle}
\declaretheorem[sharenumber=Theorem,style=bfdefstyle]{Remark}

\providecommand{\pr}[1]{\left(#1\right)} 
\providecommand{\tes}[1]{\left\lbrace\!\left\lbrace#1\right\rbrace\!\right\rbrace} 

\newcommand{\abs}[1]{\left\vert #1 \right\vert}
\newcommand{\set}[1]{\left\lbrace #1 \right\rbrace}

\renewcommand{\exp}[1]{e^{#1}}

\numberwithin{equation}{section}

\title{Generating functions for irreversible Hamiltonian systems}

\author{Dan Goreac, Jonas Kirchhoff and Bernhard Maschke}

\date{\today}

\begin{document}
\setlength{\parindent}{0em}
\pagestyle{fancy}
\lhead{Dan Goreac, Jonas Kirchhoff, Bernhard Maschke}
\rhead{Irreversible Hamiltonian systems}

\maketitle

\paragraph{Abstract} The definition of conservative-irreversible functions is extended to smooth manifolds. The local representation of these functions is studied and reveals that not each conservative-irreversible function is given by the weighted product of almost Poisson brackets. The biquadratic functions given by conservative-irreversible functions are studied and reveal a possibility for an algebraic framework on arbitrary and in particular complex algebras.


\vfill
\par\noindent\rule{5cm}{0.4pt}\\
\begin{footnotesize}
Corresponding author: Jonas Kirchhoff\\[1em]
Dan Goreac\\
Shandong University Weihai, School of Mathematics and Statistics, 180 Wenhua W Road, 264209 Weihai, Shandong, P.R. China\\
LAMA, Univ Gustave Eiffel, UPEM, Univ Paris Est Creteil, CNRS, 5 Boulevard Descartes, 77454 Marne-la-Vall{\'e}e, France\\
E-mail: dan.goreac@u-pem.fr\\[1em]
Jonas Kirchhoff\\
Institut für Mathematik, Technische Universität Ilmenau, Weimarer Stra\ss e 25, 98693 Ilmenau, Germany\\
E-mail: jonas.kirchhoff@tu-ilmenau.de\\[1em]
Bernhard Maschke\\
Lab. d’Automatique et de Genie des Procedes, Universite Claude Bernard Lyon-1, 43 Boulevard du 11 Novembre 1918, 69622 Villeurbanne Cedex, Auvergne-Rhone-Alpes, France\\
E-mail: bernhard.maschke@univ-lyon1.fr\\[1em]
Jonas Kirchhoff thanks the Technische Universität Ilmenau and the Freistaat Thüringen for their financial support as part of the Thüringer Graduiertenförderung. Bernhard Maschke acknowledges support by by IMPACTS (ANR-21-CE48-0018).
\end{footnotesize}

\newpage

\section{Introduction}

Hamiltonian systems defined on symplectic~\cite{LibeMarl87,Arno89} or, more general, Poisson manifolds~\cite{MarsRatiu99,Butt07} provide a well-established geometric theory of classical mechanics. For dynamical systems arising from irreversible thermodynamic systems, the classical Poisson geometry fails, as the effect of the \textit{entropy} of the system, which characterises reversible and irreversible processes~\cite{LiebYngva98}, has to be taken into account. To this end, several different generalisations of Poisson geometry have been proposed.

One approach, which was pioneered by~\cite{Kauf84,Morr84,Grme84} is to assume that the dissipating vector field depends, analogously to the Hamiltonian vector field, linearly on the exterior derivative of the entropy function. This results in a Leibniz bracket~\cite{OrtePlBi04}, i.e. a biderivation on $\mathcal{C}^\infty(M)$, which is usually assumed to be symmetric and non-negative, whose Leibniz vector fields with respect to the entropy model the irreversible entropy creation. To ensure that the flow of the resulting vector field satisfies the laws of thermodynamics, the non-interaction conditions that the entropy function is a Casimir of the (almost) Poisson bracket and the Hamiltonian energy function is a Casimir of the non-negative Leibniz bracket. The resulting class of systems is known as metriplectic systems~\cite{Morr86,Guha07}, or the GENERIC (\textbf{g}eneral \textbf{e}quation for \textbf{n}on-\textbf{e}quilibrium \textbf{r}eversible-\textbf{i}rreversible \textbf{c}oupling) framework~\cite{Oett05}.

A second possibility is to assume that the dissipating vector field depends, possibly non-linearly, on the Hamiltonian function~\cite{EdwaBeri91,BeriEdwa94}. Then, in addition to the noninteraction condition that the entropy is a Casimir of the (almost) Poisson bracket, the requirement that the contraction of the dissipating vector field with the Hamiltonian vanishes and with the entropy is non-negative ensures that the flow satisfies the laws of thermodynamics. Taking the free energy as the Hamiltonian function of the system, one sees that metriplectic systems can due to the noninteraction condition be viewed as a particular class of these systems. 

Yet a third possibility of extending the Poisson bracket are irreversible Hamiltonian systems, where the dissipating vector field is a Hamiltonian vector field with respect to a multiplied with the contraction of itself with the entropy and a non-negative scalar factor depending on the exterior derivative of the Hamiltonian, or a convex combination of those. The flow of the resulting vector field fulfils, provided that the entropy function is a Casimir of the original Poisson bracket, automatically the laws of thermodynamics.

In the present note, we consider a geometric structure which assigns a dissipating vector field to an arbitrary Hamiltonian function and an arbitrary entropy function so that the resulting flow automatically fulfils the laws of thermodynamics. For simplicity, we assume that this assignment is linear in the entropy and quadratic in the Hamiltonian function, and depends only on the exterior derivatives, resulting in a generalised bracket operation. This construction yields the conservative-irreversible functions as defined in~\cite{MascKirc23b}. We study the properties and demonstrate how conservative-irreversible functions may be used to combine the three approaches to thermodynamic systems in one framework.

The paper is organised as follows. In section~\ref{sec:2}, we collect a list of properties that we wish from a generalised bracket operation form thermodynamics; the resulting definition being conservative-irreversible functions. Then, following \cite{MascKirc23b}, we characterise the 4-contravariant tensor fields which are uniquely associated to these functions. In particular, we may show that a particular symmetrisation of the metriplectic 4-brackets introduced in \cite{MorrUpdi23} yields conservative-irreversible functions.

In section~\ref{sec:3}, we give the local representation of the tensor fields associated with the conservative-irreversible functions. Having identified simple conservative-irreversible functions, which are, up to scalar multiplication and arrangement of the arguments, tensor squares of almost Poisson brackets, as an important subclass, we use local representations to study the decomposability of an arbitrary conservative-irreversible function into simple functions. In particular, this implies a one-to-one correspondence between conservative-irreversible functions and metriplectic 4-brackets.

In the section~\ref{sec:4}, we characterize the biquadratic functions, which are induced by conservative-irreversible functions.

In particular, these biquadratic function define the entropy dynamics of a new class of irreversible Hamiltonian systems, which are defined using conservative-irreversible functions in section~\ref{sec:5}. We demonstrate how these systems unify, in some sense, metriplectic systems, the systems of Beris and Edwards and the original irreversible Hamiltonian systems of Ramirez et al.

\section{Conservative-irreversible functions on manifolds}\label{sec:2}


\subsection{Generalised brackets for thermodynamics}\label{sec:2.1}

We study extensions Hamiltonian systems defined on a smooth almost Poisson manifold~\cite{CantLeonDieg99} $(M,\{\cdot,\cdot\})$ where the Hamiltonian vector field is perturbed by a vector field $\Xi(s,h)$, i.e.
\begin{align}\label{eq:general_extension}
X_{s,h} := \{\cdot,h\} + \Xi(s,h),
\end{align}
for a functions $h,s\in\mathcal{C}^\infty(M)$, called the \textit{Hamiltonian} and the \textit{entropy} function. Furthermore it is assumed that the either \textit{weak noninteraction condition} $\{s,h\} = 0$ or the \textit{strong noninteraction condition} $\{s,\cdot\}\equiv 0$ holds~\cite{Morr09}. The function $\Xi:\mathcal{C}^\infty(M)\times\mathcal{C}^\infty(M)\to\mathfrak{X}(M)$, where $\mathfrak{X}(M) := \Gamma\mathcal{T}M$ denotes the space of global vector fields, is to be chosen so that the flow of the vector field~\eqref{eq:general_extension} satisfies the laws of thermodynamics, i.e.
\begin{align*}
\tfrac{\mathrm{d}}{\mathrm{d}t}s\circ x \geq 0 = \tfrac{\mathrm{d}}{\mathrm{d}t} h\circ x
\end{align*}
for each integral curve $x\in\mathcal{C}^\infty(I,M)$, $I\subseteq\mathbb{R}$ an open interval, of~\eqref{eq:general_extension}. Due to the noninteraction condition this is the case if, and only if,
\begin{align}\label{eq:thermodynamical_extension}
\iota_s \Xi(s,h) \geq 0 = \iota_{h} \Xi(s,h),
\end{align}
where $\iota$ denotes the contraction of a function with a vector field. 

In this paper, we study the geometric structure on which the functions $\Xi$ with~\eqref{eq:thermodynamical_extension} independently of the choice of the almost Poisson bracket $\{\cdot,\cdot\}$ are based on. To derive a meaningful geometric structure, we restrict ourself to ``simple'' functions $\Xi$. In particular, $\Xi$ should only depend on local data, i.e. for each $s,s',h,h'\in\mathcal{C}^\infty(M)$ and $U\subseteq M$ open with $s\vert_U - s'\vert_U = h\vert_U-h'\vert_U = 0$, we have $\Xi(s,h)\vert_U = \Xi(s',h')\vert_U$. It seems natural to assume that $\Xi$ is bilinear, thus giving a ternary bracket operation on $\mathcal{C}^\infty(M)$, and fulfil $\iota_f \Xi(g,\cdot) = \iota_g \Xi(f,\cdot)$ for all $f,g\in\mathcal{C}^\infty(M)$. Then, however, we have
\begin{align*}
\forall s,h\in\mathcal{C}^\infty(M): \iota_s \Xi(s,h)\geq 0\geq -\iota_s \Xi(s,h) = \iota_s \Xi(s,-h) \geq 0
\end{align*}
and hence $\Xi$ vanishes identically. Therefore, we instead assume that $\Xi$ is linear in the first and quadratic in the third argument. Lastly, we impose the condition that $\Xi$ is invariant under locally constant shifts of its arguments, i.e. $\Xi(s,h) = \Xi(s+c,h+c')$ for all $s,h,c,c'\in\mathcal{C}^\infty(M)$ with $\mathrm{d}c = \mathrm{d}c' = 0$, which represents the fact that energy and entropy of a given system may only be known up to some unknown (locally) constant perturbation. This defines the geometric structure of thermodynamic systems, analogously to Poisson geometry for classical mechanics. When we recall that vector fields are precisely the derivations on the algebra of smooth functions, then the functions $\Xi$ which we have described are precisely covered by the following definition, which is the direct generalisation of~\cite[Definition 3]{MascKirc23b} to manifolds.

\begin{Definition}\label{def:conservative_irreversible}
Let $M$ be a smooth manifold. We call a function $E: \big(\mathcal{C}^\infty(M)\big)^3\to \mathcal{C}^\infty(M)$ \textit{conservative-irreversible} if, and only if, there exists a function $e:\big(\mathcal{C}^\infty(M)\big)^4\to \mathcal{C}^\infty(M)$ so that
\begin{align}\label{eq:extension_CI}
\forall f,g,h\in\mathcal{C}^\infty(M): E(f,g,h) = e(f,g,h,h)
\end{align}
and, for all $f,g,h\in\mathcal{C}^\infty(M)$,
\begin{enumerate}[(i)]
\item $e$ is is a derivation, i.e. it is $\mathbb{R}$-linear and fulfils the Leibniz rule, in each argument,
\item $e(f,g,\cdot,\cdot)$ is symmetric,
\item $e(\cdot,\cdot,h,h)$ is symmetric and pointwise positive semi-definite,
\item $e(h,\cdot,h,h)$ vanishes identically,
\end{enumerate}
\end{Definition}

\begin{Remark}
Each conservative-irreversible function $E$ induces the function
\begin{align*}
\Xi_E: \mathcal{C}^\infty(M)\times\mathcal{C}^\infty(M)\to\mathfrak{X}(M),\qquad (s,h)\mapsto E(\cdot,s,h)
\end{align*}
which has exactly the properties described after~\eqref{eq:thermodynamical_extension}. Conversely, each such function $\Xi$ induces the conservative-irreversible function
\begin{align*}
E_\Xi:\big(\mathcal{C}^\infty(M)\big)^3\to \mathcal{C}^\infty(M),\qquad (f,s,h)\mapsto \iota_f\Xi(s,h);
\end{align*}
and these constructions are inverse to each other, i.e. $E_{\Xi_E} = E$ and $\Xi_{E_\Xi} = \Xi$.
\end{Remark}

\subsection{Conservative-irreversible functions}\label{sec:2.2}

In this section, we study the properties of conservative-irreversible functions. Before we do that, we ensure that the class of conservative-irreversible functions is non-trivial, when the dimension of the underlying manifold exceeds one.

\begin{Example}\label{ex:nontriv}
Let $M$ be a smooth manifold and let
\begin{align*}
\{\cdot,\cdot\}:\mathcal{C}^\infty(M)\times\mathcal{C}^\infty(M)\to\mathcal{C}^\infty(M)
\end{align*}
be a skew-symmetric biderivation, i.e. an almost Poisson bracket on $M$. Define
\begin{align*}
e:\big(\mathcal{C}^\infty(M)\big)^4\to \mathcal{C}^\infty(M),\qquad (f,s,h,g)\mapsto \frac{1}{2}\left(\{s,h\}\{f,g\}+\{f,h\}\{s,g\}\right).
\end{align*}
Since $\{\cdot,\cdot\}$ is a derivation in both arguments, $e$ is a derivation in each argument. By construction, $e(f,s,\cdot,\cdot)$ and $e(\cdot,\cdot,h,h) = \{\cdot,h\}\{\cdot,h\}$ are symmetric, and $e(s,s,h,h) = \{s,h\}^2$ implies that the latter is non-negative for all $f,s,h\in\mathcal{C}^\infty(M)$. Lastly, we observe
\begin{align*}
e(h,\cdot,h,h) = \underbrace{\{h,h\}}_{ = 0}\{\cdot,h\} = 0
\end{align*}
and hence the function
\begin{align*}
E: \big(\mathcal{C}^\infty(M)\big)^3\to \mathcal{C}^\infty(M),\qquad (f,s,h)\mapsto \{s,h\}\{f,h\}
\end{align*}
is conservative-irreversible. Finally, when the almost Poisson bracket is non-trivial, then $E$ is non-trivial; and the existence of non-trivial almost Poisson brackets in dimension at least two implies that there are non-trivial conservative-irreversible functions.
\end{Example}

\begin{Remark}
The conservative-irreversible functions constructed in Example~\ref{ex:nontriv} are closely related to the irreversible Hamiltonian systems of~\cite{RamiMascSbar13} defined on a smooth manifold $M$, whose dissipating vector field has the form
\begin{align*}
\Xi_{s,h} := (\lambda\circ\mathrm{d}h) \{s,h\}\{\cdot,h\}
\end{align*}
for a Poisson bracket $\{\cdot,\cdot\}$ and a smooth function $\lambda\in\mathcal{C}^\infty(\mathcal{T}^*M,\mathbb{R}_{>0})$.
\end{Remark}

We collect in the next lemmata some properties of the function $e$ given in~\eqref{eq:extension_CI}.

\begin{Lemma}\label{lem:Unicity_4function}
The functions $e$ of Definition~\ref{def:conservative_irreversible} are uniquely determined by their associated conservative-irreversible functions with~\eqref{eq:extension_CI}.
\end{Lemma} 
\begin{proof}
Let $e,e':\big(\mathcal{C}^\infty(M)\big)^4\to \mathcal{C}^\infty(M)$ be linear and symmetric in the third and fourth argument. Assume further that 
\begin{align*}
\forall f,g,h\in\mathcal{C}^\infty(M): e(f,g,h,h) = e'(f,g,h,h).
\end{align*}
Then, we have, for all $f,g,h,k\in\mathcal{C}^\infty(M)$,
\begin{align*}
0 & = e(f,g,h+k,h+k) - e'(f,g,h+k,h+k)\\
& = e(f,g,h,h)-e'(f,g,h,h)+e(f,g,h,k)+e(f,g,k,h)+ e(f,g,k,k)-e'(f,g,h,k)\\
& \quad -e'(f,g,k,k)-e'(f,g,k,h)\\
& = 2(e(f,g,h,k)-e'(f,g,h,k))
\end{align*}
and hence $e(f,g,h,k) = e'(f,g,h,k)$.
\end{proof}

\begin{Lemma}\label{lem:first_symmetry}
Let $e:\big(\mathcal{C}^\infty(M)\big)^4\to \mathcal{C}^\infty(M)$ be linear in each argument. Then, $e(\cdot,\cdot,h,h)$ is symmetric for all $h\in\mathcal{C}^\infty(M)$ if, and only if, $e(\cdot,\cdot,h,k)$ is symmetric for all $h,k\in\mathcal{C}^\infty(M)$.
\end{Lemma} 
\begin{proof}
The sufficiency is obvious; we show necessity. By linearity in the last two arguments, we have, for all $f,g,h,k\in\mathcal{C}^{\infty}(M)$,
\begin{align*}
e(f,g,h,k) & =\frac{1}{2}(e(f,g,h+k,h+k)-e(f,g,h,h)-e(f,g,k,k))\\
 & =\frac{1}{2}(e(g,f,h+k,h+k)-e(g,f,h,h)-e(g,f,k,k))=e(g,f,h,k).
\end{align*}
\end{proof}

In the following lemma, we characterise the tensor fields that generate conservative-irreversible functions.

\begin{Lemma}\label{lem:tensors}
A function $E: \big(\mathcal{C}^\infty(M)\big)^3\to \mathcal{C}^\infty(M)$ is conservative-irreversible if, and only if, there exists a uniquely defined contravariant tensor field $\varepsilon\in \Gamma(T^4(M))$ so that
\begin{align}\label{eq:tensorfield_representation}
\forall f,g,h\in\mathcal{C}^\infty(M)~\forall x\in M: E(f,g,h)(x) = \varepsilon_x(\mathrm{d}f_x,\mathrm{d}g_x,\mathrm{d}h_x,\mathrm{d}h_x)
\end{align}
and, for all $x\in M$, $\varepsilon_x$ has the properties
\begin{enumerate}[(a)]
\item $\varepsilon_x(\alpha,\beta,\gamma,\delta) = \varepsilon_x(\alpha,\beta,\delta,\gamma)$ for all $\alpha,\beta,\gamma,\delta\in\mathcal{T}_x^*M$,
\item $\varepsilon_x(\alpha,\beta,\gamma,\delta) + \varepsilon_x(\delta,\beta,\alpha,\gamma) + \varepsilon_x(\gamma,\beta,\delta,\alpha) = 0$ for all $\alpha,\beta,\gamma,\delta\in\mathcal{T}_x^*M$,
\item\label{it:nonnegative_tensor} $\varepsilon_x(\cdot,\cdot,\gamma,\gamma)$ is symmetric and positive semi-definite for all $\alpha\in\mathcal{T}^*_x(M)$
\end{enumerate}
\end{Lemma}
\begin{proof}

It is well-known that all $k$-derivations on $\mathcal{C}^{\infty}(M)$,
i.e. functions of $k$ argument that are linear and fulfil the Leibniz
rule in each argument, are given by $k$-contravariant tensor fields
$\xi\in\Gamma(T^{k}(M))$, where $T^{k}(M)$ denotes the vector bundle
of $k$-contravariant tensors, see e.g.~\cite[Section 3.18]{GreuHalpVans72}.
Furthermore, in view of the lemma~\ref{lem:Unicity_4function}, there
is a one-to-one correspondence between conservative-irreversible functions
and functions $e:\big(\mathcal{C}^{\infty}(M)\big)^{4}\to\mathcal{C}^{\infty}(M)$
with the properties (i)--(iv) of Definition~\ref{def:conservative_irreversible}.
Hence the relation
\begin{align*}
\forall f,g,h,k\in\mathcal{C}^\infty(M)~\forall x\in M: e(f,g,h,k)(x) = \varepsilon_x(df_x,dg_x,dh_x,dk_x)
\end{align*}
establishes a one-to-one correspondence between 4-contravariant tensor fields $\varepsilon\in\Gamma(T^4(M))$ and functions $e:\big(\mathcal{C}^\infty(M)\big)^4\to \mathcal{C}^\infty(M)$ with the property (i) of Definition~\ref{def:conservative_irreversible}. The proof of the equivalence of the properties (ii)--(iv) of Definition~\ref{def:conservative_irreversible} for $e$ and the properties (a)--(c) of the associated tensor field $\varepsilon$ is analogous to the proof of~\cite{MascKirc23b} to which we refer the reader.
\end{proof}

We show now, that the functions $e$ associated to conservative-irreversible functions by~\eqref{eq:extension_CI} are symmetric in the first and second pair of arguments.

\begin{Lemma}\label{lem:nice_symmetry}
Let $e:\left(\mathcal{C}^\infty(M)\right)^4\to\mathcal{C}^\infty(M)$ fulfil the properties (i)--(iv) of Definition~\ref{def:conservative_irreversible}. Then, $e(f,g,h,k) = e(h,k,f,g)$ for all $f,g,h,k\in\mathcal{C}^\infty(M)$.
\end{Lemma}
\begin{proof}
Let $E:\left(\mathcal{C}^\infty(M)\right)^3\to\mathcal{C}^\infty(M)$ be the conservative-irreversible function associated to $e$, and let $\varepsilon\in\Gamma(T^4(M))$ the (unique) tensor field with~\eqref{eq:tensorfield_representation}. Equivalently, we have
\begin{align*}
\forall f,g,h,k\in\mathcal{C}^\infty(M)~\forall x\in M: e(f,g,h,k)(x) = \varepsilon_x(\mathrm{d}f_x,\mathrm{d}g_x,\mathrm{d}h_x,\mathrm{d}k_x).
\end{align*}
Thus, Lemma~\ref{lem:tensors}\,(b) yields in particular
\begin{align*}
\forall f,g,h,k\in\mathcal{C}^\infty(M): e(f,g,h,k)+e(h,g,k,f)+e(k,g,f,k) = 0.
\end{align*}
Let $f,g,h,k\in\mathcal{C}^\infty(M)$. Using the symmetry of $e$ in the first two and in the last two arguments, we see
\begin{align*}
e(f,h,k,g) & = e(h,f,k,g)\\
& = -e(k,f,g,h)-e(g,f,h,k)\\
& = -e(f,k,h,g)-e(g,f,h,k)\\
& = e(h,k,g,f)+e(g,k,f,h)-e(g,f,h,k)\\
& = e(h,k,f,g)+e(k,g,f,h)-e(f,g,h,k)
\end{align*}
and therefore
\begin{align*}
0 & = e(f,h,k,g)-e(k,g,f,h)+e(f,g,h,k)-e(h,k,f,g)\\
& = -e(k,h,g,f)-e(g,h,f,k)+e(f,g,h,k)+e(h,g,k,f)+e(f,g,h,k)-e(h,k,f,g)\\
& = 2(e(f,g,h,k)-e(h,k,f,g)).
\end{align*}
This completes the proof of the Lemma.
\end{proof}

\subsection{Relation with metriplectic 4-brackets}\label{sec:2.3}

Recently, Morrison and Updike~\cite{MorrUpdi23} proposed a very similar bracket description called \textit{metriplectic four-bracket} as a geometric formulation for dissipative systems, in particular for metriplectic systems. These brackets, whose properties are motivated by Riemannian curvature tensors, are defined as follows.

\begin{Definition}[{\cite[p.\,6]{MorrUpdi23}}]
A metriplectic four-bracket on the manifold $M$ is a function $(\!(\cdot,\cdot,\cdot,\cdot)\!):\mathcal{C}^\infty(M)\times\mathcal{C}^\infty(M)\times\mathcal{C}^\infty(M)\times\mathcal{C}^\infty(M)\to\mathcal{C}^\infty(M)$ with the properties
\begin{enumerate}[(i)]
\item $(\!(\cdot,\cdot,\cdot,\cdot)\!)$ is a derivation in each argument
\item $(\!(f,g,\cdot,\cdot)\!)$ and $(\!(\cdot,\cdot,f,g)\!)$ are skew-symmetric for all $f,g\in\mathcal{C}^\infty(M)$
\item $(\!(f,g,h,k)\!) = (\!(h,k,f,g)\!)$ for all $f,g,h,k\in\mathcal{C}^\infty(M)$
\item $(\!(f,g,h,k)\!) + (\!(f,h,k,g)\!) + (\!(f,k,g,h)\!) = 0$ for all $f,g,h,k\in\mathcal{C}^\infty(M)$
\item $(\!(f,g,f,g)\!)(x)\geq 0$ for all $f,g\in\mathcal{C}^\infty(M)$ and $x\in M$
\end{enumerate}
The dissipative bracket given by $(\!(\cdot,\cdot,\cdot,\cdot)\!)$ and a Hamiltonian function $h\in\mathcal{C}^\infty(M)$ is the bracket $(\cdot,\cdot)_h := (\!(\cdot,h,\cdot,h)\!)$.
\end{Definition}

Evidently, the symmetries of conservative-irreversible tensors and metriplectic four-brackets are, at first glance, slightly different; and so are their motivations: Conservative-irreversible functions are a general class of tensor fields with easy to study symmetries so that the flow of the vector field $e(\cdot,s,h,h)$ satisfies the laws of thermodynamics for all functions $s,h\in\mathcal{C}^\infty(M)$ and each conservative-irreversible function $e$. This raises the question of the relation of metriplectic four-brackets and the four-tensors associated to conservative-irreversible functions. It is easy to see the symmetrisation of metriplectic four tensors in the second and fourth argument fulfil the conditions (i)--(iv) satisfied by the conservative-irreversible functions in the Definition~\ref{def:conservative_irreversible}.

\begin{Lemma}\label{lem:sufficiency}
Let $(\!(\cdot,\cdot,\cdot,\cdot)\!)$ be a metriplectic four-bracket on the manifold $M$. The function
\begin{align*}
e:\mathcal{C}^\infty(M)\times\mathcal{C}^\infty(M)\times\mathcal{C}^\infty(M)\times\mathcal{C}^\infty(M) & \to\mathcal{C}^\infty(M),\\
(f,g,h,k) & \mapsto \frac{1}{2}\big((\!(f,h,g,k)\!)+(\!(f,k,g,h)\!)\big)
\end{align*}
has the properties (i)--(iv) in the definition~\ref{def:conservative_irreversible}.
\end{Lemma}
\begin{proof}
Since $(\!(\cdot,\cdot,\cdot,\cdot)\!)$ is a derivation in each argument, it is readily seen that $e$ is a derivation in each argument; and $e$ is symmetric in the last two arguments by construction, so that $e$ fulfils the properties (i) and (ii) of Definition~\ref{def:conservative_irreversible}. Let $h\in\mathcal{C}^\infty(M)$. Then, the symmetry of $(\!(\cdot,\cdot,\cdot,\cdot)\!)$ in the first and the last two arguments implies, for all $f,g\in\mathcal{C}^\infty(M)$,
\begin{align*}
e(f,g,h,h) = \frac{1}{2}\big((\!(f,h,g,h)\!)+(\!(f,h,g,h)\!)\big) = \frac{1}{2}\big((\!(g,h,f,h)\!)+(\!(g,h,f,h)\!)\big) = e(g,f,h,h),
\end{align*}
so that $e(\cdot,\cdot,h,h)$ is symmetric. Furthermore, $e(f,f,h,h) = \frac{1}{2}\big((\!(f,h,f,h)\!)+(\!(f,h,f,h)\!)\big) = (\!(f,h,f,h)\!)$ is pointwise non-negative, so that $e$ fulfils (iii) of Definition~\ref{def:conservative_irreversible}. Finally, skew-symmetry of $(\!(\cdot,\cdot,\cdot,\cdot)\!)$ implies
\begin{align*}
e(h,f,h,h) = \frac{1}{2}\big((\!(h,h,f,h)\!)+(\!(h,h,f,h)\!)\big) = (\!(h,f,h,h)\!) = 0
\end{align*}
so that $e$ fulfils the property (iv) of Definition~\ref{def:conservative_irreversible}.
\end{proof}

This observation implies, in particular, that the functions $H(f,g,h):=(\!(f,h,g,h)\!)$ considered by Morrison and Updike are conservative-irreversible functions. The converse relation, i.e. for each conservative-irreversible function $H$ there exists a metriplectic four-bracket $(\!(\cdot,\cdot,\cdot,\cdot)\!)$ so that $H(f,g,h) = (\!(f,h,g,h)\!)$ for all $f,g,h\in\mathcal{C}^\infty(M)$, is not obvious.
In the special case that there is some almost Poisson bracket $\{\cdot,\cdot\}$ with $E(f,g,h) = \{g,h\}\{f,h\}$ for all $f,g,h\in\mathcal{C}^\infty(M)$, it is straightforward to verify that
\begin{align*}
\{\!\{\cdot,\cdot,\cdot,\cdot\}\!\}:\mathcal{C}^\infty(M)\times\mathcal{C}^\infty(M)\times\mathcal{C}^\infty(M)\times\mathcal{C}^\infty(M) & \to\mathcal{C}^\infty(M),\\
(f,g,h,k) & \mapsto \{f,h\}\{g,k\}
\end{align*}
is a metriplectic four-bracket, see~\cite{MorrUpdi23}. This leads to the definition of \textit{simple} conservative-irreversible functions as follows.

\begin{Definition}\label{def:simple}
We call a conservative-irreversible function $E$ \textit{simple} if, and only if, there exist an almost Poisson bracket $\{\cdot,\cdot\}$ and a non-negative function $\lambda:M\to\mathbb{R}_{\geq 0}$ so that
\begin{align}\label{eq:simple_representation}
\forall f,g,h\in\mathcal{C}^\infty(M): E(f,g,h) = \lambda\{f,h\}\{g,h\}.
\end{align}
\end{Definition}

The 4-contravariant tensor field associated to simple conservative-irreversible functions is, up to scalar multiplication and rearranging of the arguments, the tensor square of a Poisson bivector. This raises the question: Given any two 2-contravariant tensor fields, when is their tensor product (up to the same rearrangement of the arguments) the tensor field of a conservative-irreversible function? The answer is the following.

\begin{Proposition}
Let $M$ be a smooth manifold and let $[\cdot,\cdot]_1,[\cdot,\cdot]_2:\mathcal{C}^\infty(M)\times\mathcal{C}^\infty(M)\to\mathcal{C}^\infty(M)$ be biderivations. Then, the function
\begin{align}
E: \mathcal{C}^\infty(M)\times\mathcal{C}^\infty(M)\to\mathcal{C}^\infty(M),\qquad (f,g,h)\mapsto [f,h]_1[g,h]_2
\end{align}
is conservative-irreversible if, and only if, it is simple.
\end{Proposition}
\begin{proof}
Sufficiency is evident; we show necessity. Let $\Lambda^1,\Lambda^2\in\Gamma(T^2(M))$ be the tensor fields associated to $[\cdot,\cdot]_1$ and $[\cdot,\cdot]_2$, respectively. Using local coordinates, we may use the methods of the proof of~\cite[Theorem 4.2]{MascKirc23b} to see that $[\cdot,\cdot]_1$ is skew-symmetric, and hence an almost Poisson bracket. Further, we see that $\Lambda_x^1$ and $\Lambda_x^2$ are linearly dependent for all $x\in M$. Therefore, there exists a function $\lambda:M\to\mathbb{R}$ so that
\begin{align*}
\forall f,g,h,k\in\mathcal{C}^\infty(M): [f,h]_1[g,h]_2 = \lambda [f,h]_1[g,h]_1,
\end{align*}
which shows the assertion.
\end{proof}

We show that the closure of simple conservative-irreversible functions under locally finite sums is in one-to-one correspondence with the symmetrisations of metriplectic four-brackets.

\begin{Lemma}\label{cor:loc_fin_sum}
Let $E$ be a conservative-irreversible function so that there is  a family $(E_i)_{i\in I}$ of simple conservative-irreversible functions with $E = \sum_{i \in I} E_i$ so that
\begin{align*}
\forall f,g,h\in\mathcal{C}^\infty(M)~\forall K\subseteq M~\text{compact}: \abs{\set{i\in\mathbb{N}\,\big\vert\, E_i(f,g,h)\vert_K\not\equiv 0}}<\infty,
\end{align*}
i.e. the sum is locally finite. Then, there exists a metriplectic four-bracket $(\!(\cdot,\cdot,\cdot,\cdot)\!)$ so that 
\begin{align}\label{eq:mp4_1}
\forall f,g,h\in\mathcal{C}^\infty(M): E(f,g,h) = (\!(f,h,g,h)\!).
\end{align}
\end{Lemma}
\begin{proof}
Consider, for each $i\in I$, functions $\lambda_i$ and Poisson brackets $\{\cdot,\cdot\}_i$ so that $E_i(f,g,h) = \lambda_i\{f,h\}_i\{g,h\}_i$ for all $f,g,h\in\mathcal{C}^\infty(M)$. Define the functions
\begin{align*}
(\!(\cdot,\cdot,\cdot,\cdot)\!)_i:\mathcal{C}^\infty(M)\times\mathcal{C}^\infty(M)\times\mathcal{C}^\infty(M)\times\mathcal{C}^\infty(M)\to\mathcal{C}^\infty(M),(\!(f,g,h,k)\!):=\lambda_i\{f,g\}_i\{h,k\}_i.
\end{align*}
so that $(\!(f,h,g,h)\!)_i = E_i(f,g,h)$ for all $f,g,h\in\mathcal{C}^\infty(M)$. The observation
\begin{align}\label{eq:trivial_observation}
\forall f,g,h,k\in\mathcal{C}^\infty(M): (\!(f,g,h,k)\!)_i = \frac{1}{2}\left( E_i(f,h,g+k)-E_i(f,h,g)-E_i(f,h,k)\right)
\end{align}
ensures that $(\!(\cdot,\cdot,\cdot,\cdot)\!)$ is well-defined. Since $\{\cdot,\cdot\}_i$ is a Poisson bracket, $(\!(\cdot,\cdot,\cdot,\cdot)\!)$ is indeed a metriplectic four-bracket, see~\cite{MorrUpdi23}. Since the support of families $(E_i(f,g,h))_{i\in I}$ is locally finite for all $f,g,h\in\mathcal{C}^\infty(M)$, we conclude in view of~\eqref{eq:trivial_observation} that $(\!(\cdot,\cdot,\cdot,\cdot)\!) = \sum_{i\in I} (\!(\cdot,\cdot,\cdot,\cdot)\!)$ is well-defined; and since the vector space of metriplectic four-brackets is closed under arbitrary locally finite sums, $(\!(\cdot,\cdot,\cdot,\cdot)\!)$ is a metriplectic four-bracket with~\eqref{eq:mp4_1}.
\end{proof}

In Definition~\ref{def:simple}, we have not required that $\lambda$ is smooth. In the following example, we construct an almost Poisson bracket $\{\cdot,\cdot\}$, which vanishes almost nowhere, and a function $\lambda$, which is not smooth, so that the induced function $E(f,g,h) = \lambda\{f,h\}\{g,h\}$ is conservative-irreversible.

\begin{Example}
It is well-known that the function
\begin{align*}
f:\mathbb{R}\to\mathbb{R},\quad x\mapsto\begin{cases}
\exp{-\frac{1}{x^2}}, & x\neq 0,\\
0, & x = 0.
\end{cases}
\end{align*}
is smooth. Consider
\begin{align*}
\lambda:\mathbb{R}\to \mathbb{R},\quad x\mapsto \abs{x}.
\end{align*}
Of course, the restriction of $\lambda$ to $\mathbb{R}\setminus\set{0}$ is smooth and hence $\lambda f$ is smooth on $\mathbb{R}\setminus\set{0}$. Furthermore, we have
\begin{align*}
\forall x\in\mathbb{R}\setminus\set{0}: \frac{\lambda(x)f(x)-\lambda(0)f(0)}{x} = \mathrm{sgn}(x) f(x)
\end{align*}
and hence $\lambda f$ is differentiable in $0$ with vanishing derivative. Since all derivatives of $f$ in zero vanish, we find that the derivative of $\lambda f$,
\begin{align*}
\mathrm{d}(\lambda f) = \begin{cases}
\mathrm{sgn}(x) f(x) + \lambda(x) \mathrm{d}f(x), & x\neq 0,\\
0, & x = 0,
\end{cases}
\end{align*}
is indeed smooth and the $i$-th derivative of $\lambda f$, $i\in\mathbb{N}^*$, may be determined as
\begin{align*}
\mathrm{d}^i(\lambda f) = \begin{cases}
\lambda(x) d^i f(x) + i\,\mathrm{sgn}(x)\mathrm{d}^{i-1}f(x), & x\neq 0,\\
0, & x = 0.
\end{cases}
\end{align*}
Analogously, it may be shown that the product of $\lambda$ with the smooth function $g := f\circ x^2$ is smooth. Consider the smooth matrix field
\begin{align*}
J: \mathbb{R}^{3}\to\mathbb{R}^{3\times 3},\quad x\mapsto\begin{bmatrix}
0 & f(x_1) & 0\\
-f(x_1) & 0 & g(x_1)\\
0 & -g(x_1) & 0
\end{bmatrix}.
\end{align*}
Since $J$ is pointwise skew-symmetric, the associated bracket $\{\cdot,\cdot\}_J$ is almost Poisson. Since $\lambda f$ and $\lambda g$ are smooth functions, $\{\cdot,\cdot\}_{\lambda J}$ is also an almost Poisson bracket. Then, we see that
\begin{align*}
E:\mathcal{C}^\infty(\mathbb{R}^2)\times\mathcal{C}^\infty(\mathbb{R}^2)\times\mathcal{C}^\infty(\mathbb{R}^2)\to\mathcal{C}^\infty(\mathbb{R}^2),\quad (f,g,h)\mapsto \lambda\{f,h\}_J\{g,h\}_J = \{f,h\}_{\lambda J}\{g,h\}_J,
\end{align*}
is a conservative-irreversible function. In particular, we see that the function $\lambda$ in~\eqref{eq:simple_representation} does not need to be necessarily smooth.
\end{Example}

The representation~\eqref{eq:simple_representation} is, in general, not unique. Of course, we may multiply the Poisson bivector fields with $\sqrt{\frac{1}{\mu}}$, where $\mu$ is an everywhere positive smooth function, and multiply $\lambda$ with $\mu$. A different non-uniqueness of the representation~\eqref{eq:simple_representation} may be observed, if the tensor field of the conservative-irreversible function vanishes on an open set, as the next example examines.

\begin{Example}
Recall the smooth Urysohn lemma: Let $M$ be a smooth manifold and let $K,S\subseteq M$ be disjoint closed sets with $K$ compact. Then, there exists a function $u\in\mathcal{C}^\infty(M)$ so that $u\vert_K\equiv 1$ and $u\vert_S \equiv 0$. Any such function $u$ is called smooth Urysohn function. Let $\dim M\geq 2$ and $K_1,K_2\subseteq M$ be nonempty, compact sets which possess disjoint open neighbourhoods $U_1$ and $U_2$, respectively. Such sets exist in nonempty open subsets of $\mathbb{R}^{\dim M}$ and thus in $M$. Let $\varphi$ and $\psi$ be smooth Urysohn functions separating $K_1$ and $M\setminus U_1$ and $K_2$ and $M\setminus U_2$, respectively. Let $\{\cdot,\cdot\}$ be a Poisson bracket whose Poisson bivector field does not vanish on $K_1\cup K_2$ and put $[\cdot,\cdot]_1 := \varphi\{\cdot,\cdot\}$ and $[\cdot,\cdot]_2 := \psi\{\cdot,\cdot\}$. Since $\{\cdot,\cdot\}$ is a(n almost) Poisson bracket, and since the almost Poisson brackets are a sub-module of the $\mathcal{C}^\infty(M)$-module of $\mathbb{R}$-bilinear functions on $\mathcal{C}^\infty(M)$, $[\cdot,\cdot]_1,[\cdot,\cdot]_2$ are almost Poisson brackets with non-vanishing bivector field. By construction of $\varphi$ and $\psi$, we have
\begin{align*}
\forall f,g,h\in\mathcal{C}^\infty: [f,h]_1[g,h]_2 = \varphi [f,h]_2[g,h]_2 = \psi[f,h]_1[g,h]_1 \equiv 0.
\end{align*}
As a consequence, there are uncountably many representations~\eqref{eq:simple_representation} of the trivial conservative-irreversible function.
\end{Example}

\section{Local representations}\label{sec:3}

In this section, we study the representations of conservative-irreversible functions in local coordinates. As a consequence of Lemma~\ref{lem:tensors}, the local matrix fields associated to conservative-irreversible functions are readily characterised.

\begin{Proposition}\label{prop:coordinates}
A function $E:\big(\mathcal{C}^\infty(M)\big)^3\to\mathcal{C}^\infty(M)$ is conservative-irreversible if, and only if, for each patch of local coordinates $x\in \mathbb{R}^n$ on a nonempty open set $U\subseteq M$, there exists a smooth tensor field $\varepsilon:U\to\mathbb{R}^{n\times n\times n\times n}$ with the properties
\begin{enumerate}[(i)]
\item $\varepsilon_{i,j,k,\ell}(x) + \varepsilon_{k,j,\ell,i}(x) + \varepsilon_{\ell,j,i,k}(x) = 0$,
\item $\varepsilon_{i,j,k,\ell}(x) - \varepsilon_{i,j,\ell,k}(x) = 0$,
\item $\left[\sum_{k,\ell = 1}^n e_{i,j,k,\ell}(x) y_ky_\ell\right]_{i,j = 1}^n$ symmetric and positive semi-definite for all $y\in\mathbb{R}^n$
\end{enumerate}
for all $i,j,k,\ell\in\mathbb{N}^*_{\leq n} := \set{1,\ldots,n}$ and $x\in U$, so that, for all $f,g,h\in\mathcal{C}^\infty(M)$ and $x\in U$,
\begin{align}\label{eq:coordinates}
E(f,g,h)(x) = \sum_{i,j,k,\ell = 1}^n \varepsilon_{i,j,k,\ell}(x)\partial_{x_i}f(x)\partial_{x_j}g(x)\partial_{x_k}h(x)\partial_{x_\ell}h(x).
\end{align}
\end{Proposition}
\begin{proof}
Since $E$ is associated to a four-contravariant tensor field by~\eqref{eq:tensorfield_representation}, there exist smooth local tensor fields $\varepsilon$ with~\eqref{eq:coordinates}. Sufficiency and necessity of the properties (i)--(iii) for $E$ being conservative-irreversible are a direct consequence of~\cite[Proposition 8, section 3]{MascKirc23b}.
\end{proof}

Of particular interest are the simple conservative-irreversible functions, i.e. those that are (up to multiplication with a scalar factor and order of the arguments) the tensor square of an almost Poisson bracket. In the following, we want to answer the question: Are there non-simple conservative-irreversible functions? If the answer to this question is affirmative, then this raises the question of the relation between simple and general conservative-irreversible functions. In particular, contains the span of the former with respect to $\mathcal{C}^\infty(M)$ the latter or are the general conservative-irreversible functions truly different? To answer these questions, we calculate local bases of the smallest subbundle of the tensor bundle containing all conservative-irreversible functions. As a first step, we calculate the fibre-wise dimension of the vector bundle given by the properties (i) and (ii) and the symmetry of (iii) in Proposition~\ref{prop:coordinates}.

\begin{Lemma}\label{lem:ein_Lemma}
Let $n\in\mathbb{N}^*$. The dimension of the subvector space
\begin{align*}
V_3 := \set{e\in\mathbb{R}^{n\times n\times n\times n}\,\big\vert\,\forall i,j,k,\ell\in\mathbb{N}^*_{\leq n}: e_{i,j,k,\ell}+e_{k,j,\ell,i}+e_{\ell,j,i,k} = e_{i,j,k,\ell}-e_{i,j,\ell,k} = e_{i,j,k,\ell}-e_{j,i,k,\ell} = 0}
\end{align*}
of $\mathbb{R}^{n\times n\times n\times n}$ is $\frac{1}{2}n(n-1)\left(\frac{1}{6}n^2-\frac{1}{2}n+\frac{4}{3}\right)$.
\end{Lemma}
\begin{proof}
Consider the action of $S_4$, the group of permutations of $\mathbb{N}^*_{\leq 4}$, on $(\mathbb{N}^*_{\leq n})^4$ given by
\begin{align*}
\pi\cdot \alpha := (\alpha_{\pi(1)},\ldots,\alpha_{\pi(4)})
\end{align*}
for all $\alpha\in (\mathbb{N}^*_{\leq n})^4$ and $\pi\in S_4$. Recall that $S_4$ is generated by the permutations $\pi_1:=(34)$, $\pi_2 :=(12)$ and $\pi_3 := (14)(34)$. Therefore, for all $i,j,k,\ell\in\mathbb{N}_{\leq n}^*$,
\begin{align*}
S_{i,j,k,\ell} & := \set{x\in(\mathbb{N}^*_{\leq n})^4\,\big\vert\,\set{i,j,k,\ell} = \set{x_1,x_2,x_3,x_4}}\\
& = \set{\pi\cdot (i,j,k,\ell)\,\big\vert\,\pi\in S_4}\\
& = \set{\pi_{i_1}\cdots\pi_{i_k}\cdot(i,j,k,\ell)\,\big\vert\,i_1,\ldots,i_k\in\mathbb{N}_{\leq 3}^*, k\in\mathbb{N}_0}.
\end{align*}
Furthermore, the conditions
\begin{align*}
\forall i,j,k,\ell\in\mathbb{N}^*_{\leq n}:\quad e_{i,j,k,\ell}+e_{k,j,\ell,i}+e_{\ell,j,i,k} & = 0,\\
e_{i,j,k,\ell}-e_{i,j,\ell,k} & = 0,\\
e_{j,i,k,\ell}-e_{j,i,k,\ell} & = 0
\end{align*}
are equivalent to
\begin{equation}\label{eq:all_the_symmetry}
\begin{aligned}
\forall i,j,k,\ell\in\mathbb{N}^*_{\leq n}~\forall\alpha\in\set{i,j,k,\ell}^4:\quad e_{\alpha} + e_{\pi_3\cdot\alpha} & = -e_{\pi_3^2\cdot \alpha},\\
e_\alpha & = e_{\pi_2\cdot\alpha},\\
e_{\alpha} & = e_{\pi_1\cdot\alpha}.
\end{aligned}
\end{equation}
Let $i,j,k,\ell\in\mathbb{N}^*_{\leq n}$ and $e\in V_3$. Then,~\eqref{eq:all_the_symmetry} yields that the values $e_\alpha$, $\alpha\in S_{i,j,k,\ell}$, are uniquely determined by $e_{(i,j,k,\ell)}, e_{\pi_3\cdot(i,j,k,\ell)}$ and $e_{\pi_3\cdot\pi_1\cdot\pi_2\cdot(i,j,k,\ell)}$. If $i = j = k = \ell$, then~\eqref{eq:all_the_symmetry} yields that $e_\alpha = 0$ for all $\alpha\in S_{i,i,i,i} = \set{(i,i,i,i)}$. If $\abs{\set{i,j,k,\ell}} = 2$, then we may assume that $j \neq i = k = \ell$, since $S_{i,j,k,\ell}$ is invariant under permutation of $i,j,k$ and $\ell$. Then, we find
\begin{align*}
-2e_{i,i,j,i} = -e_{\pi_3\cdot(j,i,i,i)}-e_{\pi_3^2\cdot(j,i,i,i)} = e_{j,i,i,i} = e_{i,j,i,i} = \frac{1}{3}(e_{i,j,i,i}+e_{\pi_3\cdot(i,j,i,i)}+e_{\pi_3^2\cdot(i,j,i,i)}) = 0.
\end{align*}
Therefore, we conclude that $e_\alpha = 0$ for all $\alpha\in S_{i,j,k,\ell}$. If $\abs{\set{i,j,k,\ell}} = 3$, then we may without loss of generality assume that $k = \ell$. Thus,~\eqref{eq:all_the_symmetry} implies
\begin{align*}
2e_{k,i,k,j} = e_{k,i,k,j}+e_{k,i,j,k} = -e_{j,i,k,k} = -e_{i,j,k,k} = e_{k,j,k,i}+e_{k,j,i,k} = 2e_{k,j,k,i}
\end{align*}
and we conclude that the $e_\alpha$, $\alpha\in S_{i,j,k,\ell}$, are uniquely determined by $e_{i,j,k,k}$. Let $\abs{\set{i,j,k,\ell}} = 4$. Then,~\eqref{eq:all_the_symmetry} implies
\begin{align*}
0 & = -(-e_{i,j,k,\ell}-e_{\ell,i,k,j})-e_{k,j,\ell,i}+e_{\ell,i,k,j}+(-e_{i,j,k,\ell}-e_{k,j,\ell,i})\\
& = 2(e_{\ell,i,k,j}-e_{k,j,\ell,i}) = 2(e_{\pi_3\cdot(i,j,k,\ell)} - e_{\pi_3\cdot\pi_1\cdot\pi_2\cdot(i,j,k,\ell)}).
\end{align*}
Therefore, we conclude
\begin{align*}
\dim V_3 & = \binom{n}{2} + \binom{n}{3} + 2\binom{n}{4}\\
& = \frac{1}{2}n(n-1)\left(1+\frac{1}{3}(n-2)+\frac{1}{61}(n-2)(n-3)\right)\\
& = \frac{1}{2}n(n-1)\left(\frac{1}{6}n^2-\frac{1}{2}n+\frac{4}{3}\right)
\end{align*}
\end{proof}

In the case of symmetric and positive semi-definite real matrices, their linear span is the space of symmetric matrices. An analogous property can be observed for the convex subset of $V_3$ whose elements are characterised by the additional non-negativity condition of Proposition~\ref{prop:coordinates}\,(iii).

\begin{Lemma}\label{lem:dim_cons_irrev}
Let $n\in\mathbb{N}^*$. The convex cone
\begin{align*}
S := \set{e\in\mathbb{R}^{n\times n\times n\times n}\,\left\vert\,\begin{array}{l}\forall i,j,k,\ell\in\mathbb{N}^*_{\leq n}: e_{i,j,k,\ell}+e_{k,j,\ell,i}+e_{\ell,j,i,k} = e_{i,j,k,\ell}-e_{i,j,\ell,k} = e_{i,j,k,\ell}-e_{j,i,k,\ell} = 0,\\
\forall x\in\mathbb{R}^n:\left[\sum_{k,\ell = 1}^n e_{i,j,k,\ell} x_kx_\ell\right]_{i,j = 1}^n~\text{positive~semi-definite}\end{array}\right.}
\end{align*}
spans the vector space $V_3$ of Lemma~\ref{lem:ein_Lemma}
\end{Lemma}
\begin{proof}
Denote with $\mathfrak{T}_k(S) := \set{x\in S^k\,\big\vert\, x_1<\cdots < x_k}$ the set of increasing $k$-tuples in a given totally ordered set $S$. Define the functions
\begin{align*}
a: \mathfrak{T}_2(\mathbb{N}^*_{\leq n})\to V_3,\quad (a_{(\iota,\kappa)})_{i,j,k,\ell} := \begin{cases}
2, & (i,j,k,\ell)\in\set{(\iota,\iota,\kappa,\kappa),(\kappa,\kappa,\iota,\iota)},\\
-1 & (i,j,k,\ell)\in\set{\begin{array}{c}(\iota,\kappa,\kappa,\iota),(\iota,\kappa,\iota,\kappa),\\(\kappa,\iota,\iota,\kappa),(\kappa,\iota,\kappa,\iota)\end{array}},\\
0 & \text{else},
\end{cases}
\end{align*}
\begin{align*}
b: \mathfrak{T}_3(\mathbb{N}^*_{\leq n})\to V_3,\quad (b_{(\iota,\kappa,\lambda)})_{i,j,k,\ell} := \begin{cases}
2, & (i,j,k,\ell)\in\set{\begin{array}{c}(\iota,\kappa,\lambda,\lambda),(\kappa,\iota,\lambda,\lambda),\\
(\lambda,\lambda,\kappa,\iota),(\lambda,\lambda,\iota,\kappa)\end{array}},\\
-1, & (i,j,k,\ell)\in\set{\begin{array}{c}
(\lambda,\iota,\lambda,\kappa),(\lambda,\iota,\kappa,\lambda),\\
(\lambda,\kappa,\lambda,\iota),(\lambda,\kappa,\iota,\lambda),\\
(\kappa,\lambda,\lambda,\iota),(\kappa,\lambda,\iota,\lambda),\\
(\iota,\lambda,\lambda,\kappa),(\iota,\lambda,\kappa,\lambda)
\end{array}},\\
0, & \text{else},
\end{cases}
\end{align*}
\begin{align*}
c: \mathfrak{T}_4(\mathbb{N}^*_{\leq n})\to V_3,\quad (c_{(\iota,\kappa,\lambda,\mu)})_{i,j,k,\ell} := \begin{cases}
1, & (i,j,k,\ell)\in\set{\begin{array}{c}
(\iota,\kappa,\lambda,\mu),(\iota,\kappa,\mu,\lambda),\\
(\kappa,\iota,\lambda,\mu),(\kappa,\iota,\mu,\lambda),\\
(\mu,\lambda,\kappa,\iota),(\mu,\lambda,\iota,\kappa),\\
(\lambda,\mu,\kappa,\iota),(\lambda,\mu,\iota,\kappa)
\end{array}},\\
-1, & (i,j,k,\ell)\in\set{\begin{array}{c}
(\mu,\kappa,\iota,\lambda),(\mu,\kappa,\lambda,\iota),\\
(\kappa,\mu,\iota,\lambda),(\kappa,\mu,\lambda,\iota),\\
(\iota,\lambda,\mu,\kappa),(\iota,\lambda,\kappa,\mu),\\
(\lambda,\iota,\mu,\kappa),(\lambda,\iota,\kappa,\mu)
\end{array}},\\
0, & \text{else},
\end{cases}
\end{align*}
and
\begin{align*}
d: \mathfrak{T}_4(\mathbb{N}^*_{\leq n})\to V_3,\quad (d_{(\iota,\kappa,\lambda,\mu)})_{i,j,k,\ell} := \begin{cases}
1, & (i,j,k,\ell)\in\set{\begin{array}{c}
(\lambda,\kappa,\mu,\iota),(\lambda,\kappa,\iota,mu),\\
(\kappa,\lambda,\mu,\iota),(\kappa,\lambda,\iota,\mu),\\
(\mu,\iota,\lambda,\kappa),(\mu,\iota,\kappa,\lambda),\\
(\iota,\mu,\lambda,\kappa),(\iota,\mu,\kappa,\lambda)
\end{array}},\\
-1, & (i,j,k,\ell)\in\set{\begin{array}{c}
(\mu,\kappa,\iota,\lambda),(\mu,\kappa,\lambda,\iota),\\
(\kappa,\mu,\iota,\lambda),(\kappa,\mu,\lambda,\iota),\\
(\lambda,\iota,\kappa,\mu),(\lambda,\iota,\mu,\kappa),\\
(\iota,\lambda,\kappa,\mu),(\iota,\lambda,\mu,\kappa)
\end{array}},\\
0, & \text{else}.
\end{cases}
\end{align*}
Evidently, the set $B := \set{a_A,b_B,c_C,d_D\,\big\vert\,A\in\mathfrak{T}_2(\mathbb{N}^*_{\leq n}),T\in\mathfrak{P}_3(\mathbb{N}^*_{\leq n}),C,D\in\mathfrak{T}_4(\mathbb{N}^*_{\leq n})}$ is linear independent and has cardinality $\frac{1}{2}n(n-1)\left(\frac{1}{6}n^2-\frac{1}{2}n+\frac{4}{3}\right)$. Therefore, $B$ is a basis of $V_3$ and it remains to show that $B\subseteq S$. Let $(\iota,\kappa)\in\mathfrak{T}_2(\mathbb{N}^*_{\leq n})$. Then, we have, for all $x,y\in\mathbb{R}^n$,
\begin{align*}
\sum_{i,j,k,\ell = 1}^n (a_{(\iota,\kappa)})_{i,j,k,\ell}x_ix_jy_ky_\ell = 2x_\iota^2y_\kappa^2 + 2x_\kappa^2y_\iota^2-4x_\iota x_\kappa y_\iota y_\kappa = 2(x_\iota y_\kappa-x_\kappa y_\iota)^2\geq 0
\end{align*}
and hence $a_{(\iota,\kappa)}\in S$. Let $(\iota,\kappa,\lambda)\in\mathfrak{T}_3(\mathbb{N}^*_{\leq n})$ and consider $e^{(\iota,\kappa,\lambda)} := a_{(\iota,\lambda)} + a_{(\kappa,\lambda)} + b_{(\iota,\kappa,\lambda)}$. Then, we have, for all $x,y\in\mathbb{R}^n$,
\begin{align*}
\sum_{i,j,k,\ell = 1}^n e^{(\iota,\kappa,\lambda)}_{i,j,k,\ell} x_ix_jy_ky_\ell & = 2(x_\iota y_\lambda-x_\lambda y_\iota)^2 + 2(x_\kappa y_\lambda-x_\lambda y_\kappa)^2 + 4(x_\iota y_\lambda-x_\lambda y_\iota)(x_\kappa y_\lambda-x_\lambda y_\kappa)\\
& = 2(x_\iota y_\lambda -x_\lambda y_\iota + x_\kappa y_\lambda - x_\lambda y_\kappa)^2\geq 0
\end{align*}
and thus $e^{(\iota,\kappa,\lambda)}\in S$. Let $(\iota,\kappa,\lambda,\mu)\in\mathfrak{T}_4(\mathbb{N}^*_{\leq n})$ and consider $e^{(\iota,\kappa,\lambda,\mu)} := a_{(\iota,\mu)} + a_{(\kappa,\lambda)} + c_{(\iota,\kappa,\lambda,\mu)}$ and $\widetilde{e}^{(\iota,\kappa,\lambda,\mu)} := a_{(\iota,\kappa)} + a_{(\lambda,\mu)} + d_{(\iota,\kappa,\lambda,\mu)}$. Then, for all $x,y\in\mathbb{R}^n$,
\begin{align*}
\sum_{i,j,k,\ell = 1}^n e^{(\iota,\kappa,\lambda,\mu)}_{i,j,k,\ell} x_ix_jy_ky_\ell & = 2(x_\iota y_\mu-x_\mu y_\iota)^2 + 2(x_\kappa y_\lambda - x_\lambda y_\kappa)^2 + 4(x_\iota y_\mu - x_\mu y_\iota)(x_\kappa y_\lambda - x_\lambda y_\kappa)\\
& = 2(x_\iota y_\mu - x_\mu y_\iota + x_\kappa y_\lambda - x_\lambda y_\kappa)^2\geq 0
\end{align*}
and
\begin{align*}
\sum_{i,j,k,\ell = 1}^n \widetilde{e}^{(\iota,\kappa,\lambda,\mu)}_{i,j,k,\ell} x_ix_jy_ky_\ell & = 2(x_\iota y_\kappa - x_\kappa y_\iota)^2 + 2(x_\lambda y_\mu - x_\mu y_\lambda)^2 + 4(x_\kappa y_\iota - x_\iota y_\kappa)(x_\lambda y_\mu - x_\mu y_\lambda)\\
& = 2(x_\kappa y_\iota - x_\iota y_\kappa + x_\lambda y_\mu - x_\mu y_\lambda)^2\geq 0.
\end{align*}
Therefore, $e^{(\iota,\kappa,\lambda,\mu)},\widetilde{e}^{(\iota,\kappa,\lambda,\mu)}\in S$. Since $\set{a_A,e^B,e^C,\widetilde{e}^D\,\big\vert\, A\in\mathfrak{T}_2(\mathbb{N}^*_{\leq n}),B\in\mathfrak{T}_3(\mathbb{N}^*_{\leq n}),C,D\in\mathfrak{T}_4(\mathbb{N}^*_{\leq n})}$ is linear independent and $S\subseteq V_3$, we conclude
\begin{align*}
\dim V_3 = \abs{\mathfrak{T}_2(\mathbb{N}^*_{\leq n})} + \abs{\mathfrak{T}_3(\mathbb{N}^*_{\leq n})} + 2\abs{\mathfrak{T}_4(\mathbb{N}^*_{\leq n})}\leq \dim S \leq \dim V_3
\end{align*}
and hence $\dim S = \dim V_3 = \frac{1}{2}n(n-1)\left(\frac{1}{6}n^2-\frac{1}{2}n+\frac{4}{3}\right).$
\end{proof}

As a consequence of this calculation, the following result is derived.

\begin{Proposition}\label{prop:tensor_bundle}
The restriction of the vector bundle structure of $T^4M$ to the set
\begin{align*}
B := \set{\varepsilon\in T^4_xM\,\big\vert\,x\in M~\forall \alpha,\beta,\gamma,\delta\in\mathcal{T}_x^*M:\begin{array}{l}
\varepsilon(\alpha,\beta,\gamma,\delta) = \varepsilon(\beta,\alpha,\gamma,\delta) = \varepsilon(\alpha,\beta,\delta,\gamma),\\
\varepsilon(\alpha,\beta,\gamma,\delta) + \varepsilon(\gamma,\beta,\delta,\alpha)+\varepsilon(\delta,\beta,\alpha,\gamma) = 0
\end{array}}
\end{align*}
is the smallest subbundle of $T^4M$ so that each tensor field $\varepsilon$ with the properties (a)--(c) of Lemma~\ref{lem:tensors} is a smooth section of $B$.
\end{Proposition}
\begin{proof}
It is evident that $B$ is a subvector bundle of $T^4M$ and that all tensor fields associated to conservative-irreversible functions are sections of $B$. Let $B'$ denote the smallest subvector bundle of $T^4M$ with the latter property. Using an atlas of local trivialisations and a subordinate partition of unity, we see that for each $x\in M$ and each $\varepsilon\in T^4_xM$ with the properties (a)--(c) of Lemma~\ref{lem:tensors}, there exists a tensor field $\varepsilon'\in\Gamma(T^4M)$ with $\varepsilon'_x = \varepsilon$ which has the properties (a)--(c) of Lemma~\ref{lem:tensors}. Analogously to Lemma~\ref{lem:dim_cons_irrev}, we conclude that the rank of $B'$ is $\frac{1}{2}n(n-1)\left(\frac{1}{6}n^2-\frac{1}{2}n+\frac{4}{3}\right)$. With Lemma~\ref{lem:ein_Lemma}, we see that the rank of $B$ is $\frac{1}{2}n(n-1)\left(\frac{1}{6}n^2-\frac{1}{2}n+\frac{4}{3}\right)$. Since $B'\subseteq B$, we conclude that $B' = B$. This shows the proposition.
\end{proof}

We may use the basis used in the proof of Lemma~\ref{lem:dim_cons_irrev} to show that there are non-simple conservative-irreversible functions.

\begin{Proposition}\label{prop:heureka}
There exists $e\in S$ so that there is no $J\in\mathbb{R}^{n\times n}$ so that
\begin{align}\label{eq:repr}
\forall i,j,k,\ell\in\mathbb{N}^*_{\leq n}: e_{i,j,k,\ell} = J_{i,k}J_{j,\ell}+J_{i,\ell}J_{j,k}
\end{align}
if, and only if, $n\geq 3$.
\end{Proposition}
\begin{proof}
``$\implies$'' If $n = 1$, then $S = \set{0}$ and each $e$ has the representation~\eqref{eq:repr} with $J = 0$. If $n = 2$, then $V_3$ is one-dimensional and each $e\in S$ has the representation~\eqref{eq:repr} with
\begin{align*}
J = \alpha \begin{bmatrix}
0 & 1\\
-1 & 0
\end{bmatrix}
\end{align*}
for some $\alpha\in [0,\infty)$.

``$\impliedby$'' Let $n\geq 3$ and consider $e\in V_3$ defined by
\begin{align*}
\forall i,j,k,\ell\in\mathbb{N}^*_{\leq n}: e_{i,j,k,\ell} = \begin{cases}
2, & (i,j,k,\ell)\in\set{\begin{array}{c}
(1,1,2,2),(2,2,1,1),\\(1,1,3,3),(3,3,1,1)
\end{array}},\\
-1, & (i,j,k,\ell)\in\set{\begin{array}{c}
(1,2,2,1),(1,2,1,2),\\(2,1,1,2),(2,1,2,1),\\(1,3,3,1),(1,3,1,3),\\(3,1,3,1),(3,1,1,3)
\end{array}},\\
0, & \text{else}.
\end{cases}
\end{align*}
Then, we have, for all $x,y\in\mathbb{R}$,
\begin{align*}
\sum_{i,j,k,\ell = 1}^n e_{i,j,k,\ell} x_ix_jy_iy_k = 2(x_1y_2-x_2y_1)^2+2(x_1x_3-x_3x_1)^2\geq 0
\end{align*}
and hence $e\in S$. Seeking a contradiction, assume that $e$ has a representation~\eqref{eq:repr}. Then, we have
\begin{align*}
\frac{1}{2} e_{1,1,2,2} = \frac{1}{2} e_{1,1,3,3} = J_{1,2}^2 = J_{1,3}^2 = 1
\end{align*}
and thus
\begin{align*}
0 = e_{1,1,2,3} = 2J_{1,2}J_{1,3} \neq 0.
\end{align*}
This shows that $e$ does not have a representation~\eqref{eq:repr}.
\end{proof}

\begin{Corollary}
There exists a conservative-irreversible function on $\mathbb{R}^n$ which is not simple in standard coordinates if, and only if, $n\geq 2$.
\end{Corollary}
\begin{proof}
By~\cite[Definition 3.2]{MascKirc23b}, there exists for each conservative-irreversible function $E:\mathcal{C}^\infty(M)\times\mathcal{C}^\infty(M)\times\mathcal{C}^\infty(M)\to\mathcal{C}^\infty(M)$ a unique smooth $e:\mathbb{R}^n\to S$ so that
\begin{align}\label{eq:tensor_rep}
\forall f,g,h\in \mathcal{C}^\infty(M): E(f,g,h) = \sum_{i,j,k,\ell = 1}^n e_{i,j,k,\ell}\partial_i f\partial_j g\partial_k h \partial_\ell h.
\end{align}
In~\cite[Theorem 4.2]{MascKirc23b}, it was shown that $E$ is simple if, and only if, there is a pointwise skew-symmetric function $J:\mathbb{R}^n\to\mathbb{R}^{n\times n}$ so that
\begin{align*}
\forall i,j,k,\ell\in\mathbb{N}^*_{\leq n}~\forall x\in\mathbb{R}^n: e_{i,j,k,\ell}(x) = J_{i,k}(x)J_{j,\ell}(x) + J_{i,\ell}(x)J_{j,k}(x).
\end{align*}
By Proposition~\ref{prop:heureka}, we conclude that for $n\leq 2$ each conservative-irreversible function is simple. If $n \geq 3$, then there exists $e_0\in S$ that has no representation~\eqref{eq:repr}. Then, the constant function $x\mapsto e_0$ is smooth and generates by~\eqref{eq:tensor_rep} a conservative-irreversible function which is not simple.
\end{proof}

In the proof of the following Lemma, we show that each element of the basis constructed in Lemma~\ref{lem:dim_cons_irrev} is the tensor representation of a simple conservative-irreversible function. In particular, this implies the following result.

\begin{Lemma}\label{lem:local_rep}
Each conservative-irreversible function may be locally represented by the weighted sum of simple conservative-irreversible functions. 
\end{Lemma}
\begin{proof}
In the following, we use the notation introduced in the proof of Lemma~\ref{lem:dim_cons_irrev}. Let $(\iota,\kappa)\in\mathfrak{T}_{2}(\mathbb{N}^*_{\leq n})$ and define $J^{(\iota,\kappa)}\in\mathbb{R}^{n\times n}$ as
\begin{align*}
\forall i,j\in\mathbb{N}^*_{\leq n}: J_{i,j}^{(\iota,\kappa)} : \begin{cases}
1, & (i,j) = (\iota,\kappa),\\
-1, & (j,i) = (\iota,\kappa),\\
0, & \text{else}.
\end{cases}
\end{align*}
Then, we have, for all $i,j,k,\ell\in\mathbb{N}^*_{\leq n}$,
\begin{align*}
J_{i,k}^{(\iota,\kappa)}J_{j,\ell}^{(\iota,\kappa)} + J_{i,\ell}^{(\iota,\kappa)}J_{j,k}^{(\iota,\kappa)} = \left.\begin{cases}
2, & i = j = \iota~\wedge~k = \ell = \kappa,\\
-1, & (i,j,k,\ell)\in \set{\begin{array}{c}
(\iota,\kappa,\iota,\kappa),(\iota,\kappa,\kappa,\iota),\\
(\kappa,\iota,\iota,\kappa)
\end{array}},\\
0, & \text{else},
\end{cases}\right\rbrace = (\alpha_{\set{(\iota,\kappa)}})_{i,j,k,\ell}.
\end{align*}
Let $(\iota,\kappa,\lambda)\in\mathfrak{T}_{3}(\mathbb{N}^*_{\leq n})$ and define $K^{(\iota,\kappa,\lambda)}\in\mathbb{R}^{n\times n}$ by
\begin{align*}
\forall i,j\in\mathbb{N}^*_{\leq n}: K_{i,j}^{(\iota,\kappa,\lambda)} := \begin{cases}
1, & (i,j)\in\set{(\iota,\lambda),(\kappa,\lambda)},\\
-1, & (i,j)\in\set{(\lambda,\iota),(\lambda,\kappa)},\\
0, & \text{else}.
\end{cases}
\end{align*}
Then, we have, for all $i,j,k,\ell\in\mathbb{N}^*_{\leq n}$,
\begin{align*}
K_{i,k}^{(\iota,\kappa,\lambda)}K_{j,\ell}^{(\iota,\kappa,\lambda)} + K_{i,\ell}^{(\iota,\kappa,\lambda)}K_{j,k}^{(\iota,\kappa,\lambda)} = \left.\begin{cases}
2, & (i,j,k,\ell)\in\set{\begin{array}{c}
(\iota,\iota,\lambda,\lambda),(\lambda,\lambda,\iota,\iota),\\
(\kappa,\kappa,\lambda,\lambda),(\lambda,\lambda,\kappa,\kappa),\\
(\iota,\kappa,\lambda,\lambda),(\kappa,\iota,\lambda,\lambda),\\
(\lambda,\lambda,\iota,\kappa),(\lambda,\lambda,\kappa,\iota)
\end{array}},\\
-1, & (i,j,k,\ell)\in\set{\begin{array}{c}
(\iota,\lambda,\lambda,\iota),(\iota,\lambda,\iota,\lambda),\\
(\lambda,\iota,\lambda,\iota),(\lambda,\iota,\iota,\lambda),\\
(\lambda,\iota,\lambda,\kappa),(\lambda,\iota,\kappa,\lambda),\\
(\lambda,\kappa,\lambda,\iota),(\lambda,\kappa,\iota,\lambda),\\
(\kappa,\lambda,\lambda,\iota),(\kappa,\lambda,\iota,\lambda),\\
(\iota,\lambda,\lambda,\kappa),(\iota,\lambda,\kappa,\lambda)
\end{array}},\\
0, & \text{else},
\end{cases}\right\rbrace = e^{(\iota,\kappa,\lambda)}_{i,j,k,\ell}.
\end{align*}
Let lastly $(\iota,\kappa,\lambda,\mu)\in\mathfrak{T}_4(\mathbb{N}^*_{\leq n})$ and define the matrices $M^{(\iota,\kappa,\lambda,\mu)},N^{(\iota,\kappa,\lambda,\mu)}\in\mathbb{R}^{n\times n}$ with
\begin{align*}
\forall i,j\in\mathbb{N}^*_{\leq n}: M_{i,j}^{(\iota,\kappa,\lambda,\mu)} := \begin{cases}
1, & (i,j)\in\set{(\iota,\mu),(\kappa,\lambda)},\\
-1, & (i,j)\in\set{(\mu,\iota),(\lambda,\kappa)},\\
0, & \text{else},
\end{cases}
\end{align*}
and
\begin{align*}
\forall i,j\in\mathbb{N}^*_{\leq n}: N_{i,j}^{(\iota,\kappa,\lambda,\mu)} := \begin{cases}
1, & (i,j)\in\set{(\kappa,\iota),(\lambda,\mu)},\\
-1, & (i,j)\in\set{(\iota,\kappa),(\mu,\lambda)},\\
0, & \text{else}.
\end{cases}
\end{align*}
Then, we have, for all $i,j,k,\ell\in\mathbb{N}^*_{\leq n}$,
\begin{align*}
M_{i,k}^{(\iota,\kappa,\lambda,\mu)}M_{j,\ell}^{(\iota,\kappa,\lambda,\mu)} + M_{i,\ell}^{(\iota,\kappa,\lambda,\mu)}M_{j,k}^{(\iota,\kappa,\lambda,\mu)} = \left.\begin{cases}
2, & (i,j,k,\ell)\in\set{\begin{array}{c}
(\iota,\iota,\mu,\mu),(\mu,\mu,\iota,\iota),\\
(\kappa,\kappa,\lambda,\lambda),(\lambda,\lambda,\kappa,\kappa)
\end{array}},\\
1, & (i,j,k,\ell)\in\set{\begin{array}{c}
(\iota,\kappa,\lambda,\mu),(\iota,\kappa,\mu,\lambda),\\
(\kappa,\iota,\lambda,\mu),(\kappa,\iota,\mu,\lambda),\\
(\mu,\lambda,\kappa,\iota),(\mu,\lambda,\iota,\kappa),\\
(\lambda,\mu,\kappa,\iota),(\lambda,\mu,\iota,\kappa)
\end{array}},\\
-1, & (i,j,k,\ell)\in\set{\begin{array}{c}
(\mu,\kappa,\iota,\lambda),(\mu,\kappa,\lambda,\iota),\\
(\kappa,\mu,\iota,\lambda),(\kappa,\mu,\lambda,\iota),\\
(\iota,\lambda,\mu,\kappa),(\iota,\lambda,\kappa,\mu),\\
(\lambda,\iota,\mu,\kappa),(\lambda,\iota,\kappa,\mu)
\end{array}},\\
0, & \text{else},
\end{cases}\right\rbrace = e^{(\iota,\kappa,\lambda,\mu)}_{i,j,k,\ell}
\end{align*}
and
\begin{align*}
N_{i,k}^{(\iota,\kappa,\lambda,\mu)}N_{j,\ell}^{(\iota,\kappa,\lambda,\mu)} + N_{i,\ell}^{(\iota,\kappa,\lambda,\mu)}N_{j,k}^{(\iota,\kappa,\lambda,\mu)} = \left.\begin{cases}
2, & (i,j,k,\ell)\in\set{\begin{array}{c}
(\iota,\iota,\kappa,\kappa),(\kappa,\kappa,\iota,\iota),\\
(\mu,\mu,\lambda,\lambda),(\lambda,\lambda,\mu,\mu)
\end{array}},\\
1, & (i,j,k,\ell)\in\set{\begin{array}{c}
(\lambda,\kappa,\mu,\iota),(\lambda,\kappa,\iota,mu),\\
(\kappa,\lambda,\mu,\iota),(\kappa,\lambda,\iota,\mu),\\
(\mu,\iota,\lambda,\kappa),(\mu,\iota,\kappa,\lambda),\\
(\iota,\mu,\lambda,\kappa),(\iota,\mu,\kappa,\lambda)
\end{array}},\\
-1, & (i,j,k,\ell)\in\set{\begin{array}{c}
(\mu,\kappa,\iota,\lambda),(\mu,\kappa,\lambda,\iota),\\
(\kappa,\mu,\iota,\lambda),(\kappa,\mu,\lambda,\iota),\\
(\lambda,\iota,\kappa,\mu),(\lambda,\iota,\mu,\kappa),\\
(\iota,\lambda,\kappa,\mu),(\iota,\lambda,\mu,\kappa)
\end{array}},\\
0, & \text{else}.
\end{cases}\right\rbrace = \widetilde{e}^{(\iota,\kappa,\lambda,\mu)}_{i,j,k,\ell}.
\end{align*}
In the proof of Lemma~\ref{lem:dim_cons_irrev}, we have seen that 
\begin{align*}
B = \set{a_A,e^B,e^C,\widetilde{e}^D\,\big\vert\, A\in\mathfrak{T}_2(\mathbb{N}^*_{\leq n}),B\in\mathfrak{T}_3(\mathbb{N}^*_{\leq n}),C,D\in\mathfrak{T}_4(\mathbb{N}^*_{\leq n})}
\end{align*}
is a basis of $V_3$. Let $E$ be a conservative-irreversible function with associated representation~\eqref{eq:tensor_rep} for some $e\in \mathcal{C}^\infty(\mathbb{R}^n,S)$. Since $B$ is a basis of $V_3$, there exist $\alpha_A,\beta_B,\gamma_C,\delta_D\in\mathcal{C}^\infty(\mathbb{R}^n,\mathbb{R})$, $A\in\mathfrak{T}_2(\mathbb{N}^*_{\leq n})$, $B\in\mathfrak{T}_3(\mathbb{N}^*_{\leq n})$, $C,D\in\mathfrak{T}_4(\mathbb{N}^*_{\leq n})$, so that
\begin{align}\label{eq:1}
e = \sum_{A\in\mathfrak{T}_2(\mathbb{N}^*_{\leq n})} \alpha_A a_A + \sum_{B\in\mathfrak{T}_3(\mathbb{N}^*_{\leq n})}\beta_B e^B + \sum_{C\in\mathfrak{T}_4(\mathbb{N}^*_{\leq n})} \gamma_C e^C + \sum_{D\in\mathfrak{T}_4(\mathbb{N}^*_{\leq n})}\delta_D \widetilde{e}^D
\end{align}
For each skew-symmetric matrix $J\in\mathbb{R}^{n\times n}$, we denote with
\begin{align*}
\{\cdot,\cdot\}_J := (f,g)\mapsto (\partial f)^\top J (\partial g)
\end{align*}
the Poisson bracket induced by $J$. Then, the representation~\eqref{eq:1} and our calculations above yield
\begin{align*}
\forall f,g,h\in\mathbb{R}^n: E(f,g,h) & = \sum_{A\in\mathfrak{T}_2(\mathbb{N}^*_{\leq n})} 2\alpha_A \{f,h\}_{J^A}\{g,h\}_{J^A} + \sum_{B\in\mathfrak{T}_3(\mathbb{N}^*_{\leq n})}2\beta_B \{f,h\}_{K^B}\{g,h\}_{K^B}\\
& \quad + \sum_{C\in\mathfrak{T}_4(\mathbb{N}^*_{\leq n})} 2\gamma_C \{f,h\}_{M^C}\{g,h\}_{M^C} + \sum_{D\in\mathfrak{T}_4(\mathbb{N}^*_{\leq n})}2\delta_D \{f,h\}_{N^D}\{g,h\}_{N^D}.
\end{align*}
This shows that $E$ is indeed the weighted sum of simple conservative-irreversible functions.
\end{proof}

We have shown that any conservative-irreversible function is the weighted sum of simple conservative-irreversible functions. One might expect that all scalar factors are non-negative. This is, however, not necessarily the case, as the following example shows.

\begin{Example}
 We show that the weights constructed in the proof do not necessarily need to be positive. Let $n\geq 3$ and $(i,j,k)\in\mathfrak{T}_{3}(\mathbb{N}^*_{\leq n})$. Consider $e := a_{(i,k)} + a_{(j,k)} - \frac{1}{2} e^{(i,j,k)}$. Then, we have, for all $x,y\in\mathbb{R}^n$,
\begin{align*}
f(x,y) := \frac{1}{2}\sum_{i,j,k,\ell = 1}^n e_{i,j,k,\ell}x_ix_jy_iy_j & = (x_iy_k-x_ky_i)^2 + (x_jy_k-x_ky_j)^2\\
& \quad -\frac{1}{2}((x_iy_k-x_ky_i)+(x_jy_k-x_ky_j))^2\\
& = \frac{1}{2}(x_iy_k-x_ky_i)^2 + \frac{1}{2}(x_jy_k-x_ky_j)^2-(x_iy_k-x_ky_i)(x_jy_k-x_ky_j)\\
& = \frac{1}{2}((x_iy_k-x_ky_i)-(x_jy_k-x_ky_j))^2\geq 0.
\end{align*}
Therefore, $e\in S$ and the associated function
\begin{align*}
E(f,g,h) := \sum_{i,j,k,\ell} e_{i,j,k,\ell}\partial_i f\partial_j g\partial_k h\partial_\ell h
\end{align*}
is conservative-irreversible.
\end{Example}

As a consequence of Lemma~\ref{lem:local_rep}, we shall now conclude that each conservative-irreversible function is generated by a metriplectic four-bracket. In particular, the tensor fields which are uniquely associated to conservative-irreversible functions are precisely the symmetrisations of metriplectic four-brackets in the second and third argument.

\begin{Proposition}
A function $E:\mathcal{C}^\infty(M)\times \mathcal{C}^\infty(M)\times\mathcal{C}^\infty(M)\to\mathcal{C}^\infty(M)$ is a conservative-irreversible function if, and only if, there exists a metriplectic four-bracket $(\!(\cdot,\cdot,\cdot,\cdot)\!)$ on $M$ with
\begin{align}\label{eq:mp4_2}
\forall f,g,h\in\mathcal{C}^\infty(M): E(f,g,h) = (\!(f,h,g,h)\!)
\end{align}
\end{Proposition}
\begin{proof}
Sufficiency is a direct consequence of Lemma~\ref{lem:sufficiency}; we show the necessity. By Lemma~\ref{lem:local_rep}, there exists an open covering $(U_i)_{i\in I}$ of $M$ and a family $((E_{i,1},\ldots,E_{i,k_i}))_{i\in I}$ of simple conservative-irreversible functions on $U_i$ with
\begin{align*}
\forall f,g,h\in\mathcal{C}^\infty(M)~\forall i\in I: E(f,g,h)\vert_{U_i} = E_{i,1}(f,g,h)\vert_{U_i} +\cdots +E_{i,k_i}(f,g,h)\vert_{U_i}.
\end{align*}
Consider a (locally finite) partition of unity $(\varphi_j)_{j\in J}$ subordinate to the covering $(U_i)_{i\in I}$ and define $\overline{E}_{i,\ell}$ as the continuation of $E_{i,j}$ on $M$ by zero, $i\in I$, $\ell\in\set{1,\ldots,k_i}$. Choose, for each $j\in J$, $\iota(j)\in I$ with $\mathrm{supp}\,\varphi_j\subseteq U_{\iota(j)}$. Then, we have $E = \sum_{j\in J}\varphi_j \sum_{\ell = 1}^{k_{\iota(j)}}\overline{E}_{\iota(j),\ell}$ and the sum is a locally finite sum of simple conservative-irreversible functions. In view of Lemma~\ref{cor:loc_fin_sum}, $E$ has the representation~\eqref{eq:mp4_2} for a metriplectic four-bracket.
\end{proof}

\section{Biquadratic functions}\label{sec:4}

Let $E:\big(\mathcal{C}^\infty(M)\big)^3\to\mathcal{C}^\infty(M)$ be a conservative-irreversible function. Due to Lemma~\ref{lem:first_symmetry}, $E$ is symmetric in the first two entries and induces therefore the function 
\begin{align*}
\tes{\cdot,\cdot}_E:\mathcal{C}^\infty(M)\times\mathcal{C}^\infty(M)\to\mathcal{C}^\infty(M),\qquad (f,g)\mapsto E(f,f,g),
\end{align*}
which is \textit{biquadratic}, i.e.~a quadratic function (see Definition~\ref{def:quadr_fun}) in each entry. Conversely, to each biquadratic function $\tes{\cdot,\cdot}$, there is a four-linear function (unique, if symmetry in the first two and second two entries is imposed) $\varepsilon$ with $\tes{f,g} = \varepsilon(f,f,g,g)$. In this section, we characterise the quadratic functions whose uniquely associated symmetric bilinear functions are derivations in each argument. The proofs use general properties of algebras so that this characterisation is not only valid for the algebra of smooth real-valued functions, but for any unital commutative algebra over a field $k$ with $1_k+1_k\neq 0_k$, i.e. the characteristic is not two. As a corollary, we obtain a characterisation of the biquadratic functions associated to conservative-irreversible functions. Recall the elementary definition of quadratic functions.

\begin{Definition}[{see~\cite{HelmMicaRevo14}}]\label{def:quadr_fun}
Let $V,W$ be vector spaces over the same field $k$. A function $q:V\to V$ is \textit{quadratic} if, and only if, $q$ is $k$-homogeneous of degree two, i.e.
\begin{align*}
\forall \lambda\in k~\forall x\in V: q(\lambda x) = \lambda^2 q(x),
\end{align*}
and the function
\begin{align}\label{eq:associated_bilinear}
\beta: V\times V\to W,\qquad (x,y)\mapsto q(x+y)-q(x)-q(y)
\end{align}
is bilinear.
\end{Definition}

When the field $k$ is fixed and there is no danger of confusion, we will refer to homogeneity as degree-two-homogeneity.

\begin{Remark}
Unless $k$ is a field of characteristic two, it is easy to see that a function $q:V\to W$ is quadratic if, and only if, there exists a bilinear function $b:V\times V\to W$ so that $q(x) = b(x,x)$ for all $x\in V$: Indeed, if $q$ is quadratic, then $\beta$ given in~\eqref{eq:associated_bilinear} fulfils
\begin{align*}
\beta(x,x) = q(x+x)-q(x)-q(x) = ((1_k+1_k)^2-(1_k+1_k))q(x) = (1_k+1_k)q(x);
\end{align*}
and the converse implication is straightforward. In characteristic two, however, $1_k+1_k = 0_k$ and hence $q$ cannot be reconstructed as the concatenation of the associated bilinear function $\beta$ and the diagonal embedding of $V$ into $V\times V$. For details about the theory of quadratic function in characteristic two, see e.g.~\cite{ElmaKarpMerk08,Kneb10}.
\end{Remark}

Evidently, not each degree-two-homogeneous function is quadratic: all squares of norms are homogeneous of degree two, but not all norms are induced by inner products. A well-known characterisation of quadratic forms is the following.

\begin{Lemma}[{see~\cite[Theorem 1.1]{HelmMicaRevo14}}]
Let $V$, $W$ be vector spaces over the same field $k$. A function $q:V\to W$ is quadratic if, and only if, $q$ fulfils the inclusion-exclusion principle
\begin{align}\label{eq:inclusion_exclusion}
\delta(x+y+z)-\delta(x+y)-\delta(x+z)-\delta(y+z)+\delta(x)+\delta(x)+\delta(y)=0,
\end{align}
and the generalized binomial equation
\begin{align}\label{eq:gen_binom}
\forall x,y\in V~\forall\lambda\in k: q(\lambda x+y) = \lambda^2 q(x)+\lambda(q(x+y)-q(x)-q(y)) + q(y).
\end{align}
\end{Lemma}

A combination of the binomial formula with the Leibniz formula characterises quadratic functions whose associated bilinear functions are biderivations in characteristic unequal to two.

\begin{Proposition}\label{prop:algebraic_interlude_1}
Let $A$ be a unital commutative algebra with identity $\mathbf{1}$ over the field $k$ with characteristic not two. A function $\delta:A\to A$ is quadratic and the associated bilinear function
\begin{align*}
\Delta:A\times A\to A,\qquad (x,y)\mapsto\delta(x+y)-\delta(x)-\delta(y)
\end{align*}
is a biderivation if, and only if, $\delta$ fulfils the inclusion-exclusion principle~\eqref{eq:inclusion_exclusion} and, for all $x,y,z\in A$ and $\lambda\in k$,
\begin{equation}\label{eq:interessant_2}
\begin{aligned}
\delta(xy+\lambda z) & = xy\delta(x+y)+\lambda x\delta(y+z)+\lambda y\delta(x+z)+\lambda(\lambda\cdot\mathbf{1}-x-y)\delta(z)\\
& \quad+y(y-x-\lambda\cdot\mathbf{1})\delta(x)+x(x-y-\lambda\cdot\mathbf{1})\delta(y).
\end{aligned}
\end{equation}
\end{Proposition}
\begin{proof}
\noindent ``$\implies$'' Due to~\cite[Theorem 1.1]{HelmMicaRevo14}, $\delta$ fulfils the inclusion-exclusion principle~\eqref{eq:inclusion_exclusion}. Since $2_k:=1_k+1_k\neq 0_k$, the associated bilinear function $\Delta$ fulfils $q(x) = \frac{1}{2_k}\Delta(x,x)$ for all $x\in A$. Due to bilinearity, Leibniz rule and symmetry, we have
\begin{align*}
\delta(xy+z) & = \frac{1}{2_k}\Delta(xy+z,xy+z)\\
& = \frac{1}{2_k}\big(\Delta(xy,xy)+2\Delta(xy,z)+\Delta(z,z)\big)\\
& = \frac{1}{2_k}\big(x^2\Delta(y,y)+2xy\Delta(x,y)+y^2\Delta(x,x) + 2x\Delta(y,z) + 2y\Delta(x,z) +\Delta(z,z)\big)\\
& = \frac{1}{2_k}\big(x^2\Delta(y,y)+y^2\Delta(x,x)+\Delta(z,z) + xy(\Delta(x+y,x+y)-\Delta(x,x)-\Delta(y,y))\\
&\quad + y(\Delta(x+z,x+z)-\Delta(x,x)-\Delta(z,z)) + x(\Delta(y+z,y+z)-\Delta(y,y)-\Delta(z,z))\big)\\
& = xy\delta(x+y) + y\delta(x+z)+x\delta(y+z) + (1-y-x)\delta(z) + x(x-y-1)\delta(y)\\
&\quad + y(y-x-1)\delta(x).
\end{align*}
and hence $\delta$ fulfils indeed~\eqref{eq:interessant_2}.

\noindent ``$\impliedby$'' We proceed in steps.

\noindent\textsc{Step 1:} We show that $\delta(\mathbf{1}) = \delta(\mathbf{1}+\mathbf{1}) = 0$. Putting $x = y = z = \mathbf{1}$ and $\lambda = 0_k$ in~\eqref{eq:interessant_2} yields
\begin{align*}
\delta(\mathbf{1}) = \delta(\mathbf{1}\cdot\mathbf{1} + 0_k\cdot\mathbf{1}) = \delta(\mathbf{1}+\mathbf{1}).
\end{align*}
Therefore, putting $x = y = z = \mathbf{1}$ and $\lambda = 1_k$ in~\eqref{eq:interessant_2} yields
\begin{align*}
\delta(\mathbf{1}) = \delta(\mathbf{1}+\mathbf{1}) & = \delta(\mathbf{1}+\mathbf{1}) + \delta(\mathbf{1}+\mathbf{1})+\delta(\mathbf{1}+\mathbf{1})+1_k\cdot(\mathbf{1}-\mathbf{1}-\mathbf{1})\delta(\mathbf{1})+\mathbf{1}(\mathbf{1}-\mathbf{1}-\mathbf{1})\delta(1)\\
&\quad +\mathbf{1}(\mathbf{1}-\mathbf{1}-\mathbf{1})\delta(1)\\
& = \delta(\mathbf{1})+\delta(\mathbf{1})+\delta(\mathbf{1})-\delta(\mathbf{1})-\delta(\mathbf{1})-\delta(\mathbf{1}) = 0.
\end{align*}

\noindent \textsc{Step 2:} We show that $\delta(x+\mathbf{1}) = \delta(x)$ for all $x\in A$. Let $x\in A$. Then, we have
\begin{align*}
\delta(x+\mathbf{1}) & = \delta(\mathbf{1}\cdot\mathbf{1}+1_k\cdot x)\\
& = \delta(\mathbf{1}+\mathbf{1})+\delta(x+\mathbf{1})+\delta(x+\mathbf{1})-\delta(x)-\delta(\mathbf{1})-\delta(\mathbf{1})\\
& = \delta(x+\mathbf{1})+\delta(x+\mathbf{1})-\delta(x).
\end{align*}
This implies $\delta(x) = \delta(x+\mathbf{1})$.

\noindent \textsc{Step 3:} We show that
\begin{align*}
\forall x,y\in A~\forall\lambda\in k: \delta(\lambda x+y) = \lambda^2 \delta(x)+\lambda(\delta(x+y)-\delta(x)-\delta(y)) + \delta(y).
\end{align*}
Let $x,y\in A$ and $\lambda\in k$. Then,
\begin{align*}
\delta(\lambda x + y) & = \delta(\mathbf{1}\cdot y + \lambda x)\\
& = y\delta(y+\mathbf{1})+\lambda y\delta(x+\mathbf{1})+\lambda\delta(x+y)+\lambda(\lambda\cdot\mathbf{1}-y-\mathbf{1})\delta(x)+(\mathbf{1}-y-\lambda\cdot\mathbf{1})\delta(y)\\
& \quad+y(y-x-\lambda\cdot\mathbf{1})\delta(\mathbf{1})\\
& = y\delta(y)+\lambda y\delta(x)+\lambda\delta(x+y)+\lambda^2\delta(x)-\lambda y\delta(x)-\lambda\delta(x)+\delta(y)-y\delta(y)-\lambda\delta(y)\\
& = \lambda^2 \delta(x)+\lambda(\delta(x+y)-\delta(x)-\delta(y)) + \delta(y)
\end{align*}

\noindent \textsc{Step 4:} We have verified that $\delta$ fulfils the generalised binomial formula~\eqref{eq:gen_binom}. Therefore,~\cite[Theorem 1.1]{HelmMicaRevo14} implies that $\delta$ is a quadratic function. In particular, the associated function $\Delta$ is bilinear.

\textsc{Step 5:} We show that $\Delta$ is a derivation in each argument. Due to linearity and symmetry, it suffices to show that $\Delta$ fulfils the Leibniz rule in the first argument. Let $x,x',y\in A$. Then, we have
\begin{align*}
\Delta(xx',y) & = \delta(xx'+y)-\delta(xx')-\delta(y)\\
& = xx'\delta(x+x')+x\delta(x'+y)+x'\delta(x+y)+(\mathbf{1}-x-x')\delta(y)+x'(x'-x-\mathbf{1})\delta(x)\\
& \quad +x(x-x'-\mathbf{1})\delta(x')-xx'\delta(x+x')-x\delta(x')-x'\delta(x)-(\mathbf{1}-x-x')\delta(0)\\
& \quad -x'(x'-x-\mathbf{1})\delta(x)-x(x-x'-\mathbf{1})\delta(x')-\delta(y)\\
& = x\delta(x'+y)-x\delta(x')-x\delta(y) + x'\delta(x+y)-x'\delta(x)-x'\delta(y)\\
& = x\Delta(x',y) + x'\Delta(x,y).
\end{align*}
This completes the proof of the proposition.
\end{proof}

In characteristic two, however, Proposition~\ref{prop:algebraic_interlude_1} remains no longer true, as the following proposition shows.

\begin{Proposition}\label{eq:char_two_never_disappoints}
Let $A$ be a unital commutative algebra with identity $\mathbf{1}$ over the field $k$ with characteristic two. A function $q:A\to A$ fulfils~\eqref{eq:interessant_2}, if, and only if, $q\equiv 0$.
\end{Proposition}
\begin{proof}
Evidently, if $q\equiv 0$, then~\eqref{eq:interessant_2} is fulfilled. Conversely, let $q$ fulfil~\eqref{eq:interessant_2}. Putting $x = y = z = 0$ and $\lambda = 0_k$, we see that $q(0) = 0$. Using the calculation of Step 5 in the proof of Proposition~\ref{prop:algebraic_interlude_1}, we see that $(x,y)\mapsto q(x+y)-q(x)-q(y)$ is a derivation in the first argument. Therefore, we find, for all $x,y\in A$,
\begin{align*}
0 & = q(xy+xy)-q(xy)-q(xy) = x(q(xy+y)+q(y)+q(xy))+y(q(xy+x)+q(x)+q(xy)).
\end{align*}
Putting $x = 1$ and rearranging the terms using $a = -a$ for all $a\in A$, we conclude
\begin{align*}
\forall y\in A: q(y) = y(q(y+\mathbf{1})-q(y)-q(\mathbf{1})).
\end{align*}
In particular, we find
\begin{align*}
q(\mathbf{1}) = \mathbf{1}(q(\mathbf{1}+\mathbf{1})-q(\mathbf{1})-q(\mathbf{1})) = q(\mathbf{1}+\mathbf{1}) = q(0) = 0.
\end{align*}
Therefore, we have, for all $y\in A$,
\begin{align*}
q(y) = q(\mathbf{1}y) = y(q(y+\mathbf{1})-q(y)-q(\mathbf{1}))+q(y)+y^2q(\mathbf{1}) = q(y)+q(y) = 0.
\end{align*}
\end{proof}

In particular, Proposition~\eqref{eq:char_two_never_disappoints} yields that the implication ``$\impliedby$'' of Proposition~\ref{prop:algebraic_interlude_1} holds true in characteristic zero. This can also be seen from the observation that we have made no argument that requires $1_k+1_k\neq 0_k$. In the following example, we construct a non-zero quadratic function $\delta:\mathbb{F}_2[X]\to\mathbb{F}_2[X]$, which is the concatenation of a biderivation with the diagonal embedding $\mathbb{F}_2[X]\hookrightarrow F_2[X]\times\mathbb{F}_2[X]$. In particular, this shows that the implication ``$\implies$'' of Proposition~\ref{prop:algebraic_interlude_1} fails in characteristic two.

\begin{Example}
Consider the ring of polynomials $\mathbb{F}_2[X] = \set{p\in\mathbb{F}_2^{\mathbb{N}_0}\,\left\vert\,\sup\set{i\in\mathbb{N}_0\,\big\vert\,p_i\neq 0}<\infty\right.}$ as a commutative, unital algebra over $\mathbb{F}_2$. Since the monomials $X^n$, $n\in\mathbb{N}_0$ are a basis of $\mathbb{F}_2[X]$, there is a unique linear function
\begin{align*}
d: \mathbb{F}_2[X]\to\mathbb{F}_2[X],\qquad X^n\mapsto (n\!\!\mod 2)X^{n-1},\quad n\in\mathbb{N}_0.
\end{align*}
This function fulfils, for all $n,m\in\mathbb{N}_0$,
\begin{align*}
X^nd(X^m)+X^m d(X^n) & = (m\!\!\mod 2)X^{n}X^{m-1}+(n\!\!\mod 2)X^{m}X^{n-1}\\
& = ((m+n)\!\!\mod 2)X^{n+m-1} = d(X^{n+m}) = d(X^nX^m)
\end{align*}
and hence $d$ is a derivation. In particular, the function
\begin{align*}
\widetilde{d}: \mathbb{F}_2[X]\times\mathbb{F}_2[X]\to\mathbb{F}_2[X],\qquad (x,y)\mapsto d(x)d(y)
\end{align*}
is a biderivation. The associated quadratic function $q = d(\cdot)^2$ does not vanish, and the associated bilinear function $\Delta(x,y) = q(x+y)-q(x)-q(y)$ vanishes for all $x,y\in \mathbb{F}_2[X]$ and is therefore a biderivation.
\end{Example}

We are now able to characterise the biquadratic functions induced by conservative-irreversible functions. To this end, we introduce the following objects.

\begin{Definition}\label{def:SQPS}
Let $M$ be a smooth manifold. A function $\tes{\cdot,\cdot}:\mathcal{C}^\infty\pr{M}\times\mathcal{C}^\infty\pr{M}\to\mathcal{C}^\infty\pr{M}$ is a \emph{square quasi-Poisson structure} (SQPS) if, and only if, for each $h\in\mathcal{C}^\infty\pr{M}$, the functions ${}_h\delta := \tes{h,\cdot}$ and $\delta_h := \tes{\cdot,h}$ fulfil
\begin{enumerate}[(i)]
\item $\delta_h(f), {}_h\delta(f)$ are pointwise non-negative for all $f\in\mathcal{C}^\infty\pr{M}$,
\item $\delta_h(\cdot)$ and ${}_h\delta(\cdot)$ are \emph{$h$-translation invariant}, i.e. $\delta_h(\cdot+h) = \delta_h$ and ${}_h\delta(\cdot+h) = {}_h\delta(\cdot)$, for all $h\in\mathcal{C}^\infty(M)$,
\item $\delta_h$ and ${}_h\delta$ fulfil the inclusion-exclusion principle~\eqref{eq:inclusion_exclusion},
\item $\delta_h$ and ${}_h\delta$ fulfil~\eqref{eq:interessant_2}.
\end{enumerate}
If (only) $\delta_h$  (respectively ${}_h\delta$) is required to be $h$-translation invariant, then we call $\tes{\cdot,\cdot}$ \textit{left} (respectively \textit{right}) SQPS.
\end{Definition}

Now, we see that conservative-irreversible functions are uniquely given by left SQPS.

\begin{Proposition}\label{prop:equivalence_SQPS_CI}
A function $E:\mathcal{C}^\infty(M)\times\mathcal{C}^\infty(M)\times\mathcal{C}^\infty(M)$ is a conservative-irreversible function if, and only if, there is a left SQPS $\tes{\cdot,\cdot}$ so that
\begin{align}\label{eq:SQPS_representation}
\forall f,g,h\in\mathcal{C}^\infty(M): E(f,g,h) = \frac{1}{2}(\{\!\{f+g,h\}\!\}-\tes{f,h}-\tes{g,h});
\end{align}
and this SQPS is unique.
\end{Proposition}
\begin{proof}
\noindent ``$\implies$''\quad Let $E$ be a conservative-irreversible function and define
\begin{align*}
\tes{\cdot,\cdot}:\mathcal{C}^\infty(M)\times\mathcal{C}^\infty(M)\to\mathcal{C}^\infty(M),\quad (f,g)\mapsto E(f,f,g).
\end{align*}
Let $e$ be the unique representation~\eqref{eq:extension_CI} of $E$. Due to linearity and symmetry in the first two entries of $e$ and hence in the first two entries of $E$, we have, for all $f,g,h\in\mathcal{C}^\infty(M)$
\begin{align*}
E(f,g,h) = \frac{1}{2}(E(f+g,f+g,h)-E(f,f,h)-E(g,g,h)) = \frac{1}{2}(\tes{f+g,h}-\tes{f,h}-\tes{g,h})
\end{align*}
and hence the function $\tes{\cdot,\cdot}$ fulfils~\eqref{eq:SQPS_representation}. It remains to show that $\tes{\cdot,\cdot}$ is a left SQPS. From Definition~\ref{def:conservative_irreversible}\,(iii), we conclude
\begin{align*}
\forall f,g\in\mathcal{C}^\infty(M)~\forall x\in M: \tes{f,g}(x) = E(f,f,g)(x) = e(f,f,g,g)(x)\geq 0.
\end{align*}
and hence $\tes{\cdot,\cdot}$ has the property (i) of Definition~\ref{def:SQPS}. Let $f,h\in\mathcal{C}^\infty(M)$. In view of Definition~\ref{def:conservative_irreversible}\,(iv) and the symmetry of $E$ in the first two entries, we have
\begin{align*}
\tes{f+h,h} & = E(f+h,f+h,h) = E(f,f,h) + E(f,h,h) + E(h,f,h) + E(h,h,h)\\
& = E(f,f,h) = \tes{f,h}
\end{align*}
and hence $\tes{\cdot,h}$ is $h$-translation invariant. Further, Proposition~\ref{prop:algebraic_interlude_1} implies that $\tes{\cdot,h}$ fulfils Definition~\ref{def:SQPS}\,(iii) and (iv). This shows that $\tes{\cdot,\cdot}$ is indeed a left SQPS.

\noindent ``$\impliedby$'' Let $\tes{\cdot,\cdot}$ be a left SQPS with~\eqref{eq:SQPS_representation}. Define
\begin{align*}
e:\left(\mathcal{C}^\infty(M)\right)^4 & \to\mathcal{C}^\infty(M),\\
(f,g,h,k) & \mapsto \frac{1}{4}\big(\tes{f+g,h+k}+\tes{f,h}+\tes{f,k}+\tes{g,h}+\tes{g,k}\\
& \hspace{1cm}-\tes{f+g,h}-\tes{f+g,k}-\tes{f,h+k}-\tes{g,h+k}\big).
\end{align*}
Due to Proposition~\ref{prop:algebraic_interlude_1}, we have, for all $f,g,h\in \mathcal{C}^\infty(M)$,
\begin{align*}
E(f,g,h) & = \frac{1}{2}(\tes{f+g,h}-\tes{f,h}-\tes{g,h})\\
& = \frac{1}{4}\big(4\tes{f+g,h}+2\tes{f,h}+2\tes{g,h}-2\tes{f+g,h}-4\tes{f,h}-4\tes{g,h}\big)\\
& = \frac{1}{4}\big(\tes{f+g,2h}+2\tes{f,h}+2\tes{g,h}-2\tes{f+g,h}-\tes{f,2h}-\tes{g,2h}\big)\\
& = e(f,g,h,h)
\end{align*}
and hence $e$ generates $E$ with~\eqref{eq:extension_CI}. It remains to verify that $e$ fulfils Definition~\ref{def:conservative_irreversible}\,(i)--(iv). The symmetry of $e$ in the first two arguments and in the last two arguments is a direct consequence of the commutativity of the addition in $\mathcal{C}^\infty(M)$. In view of Definition~\ref{def:SQPS}\,(i) and Proposition~\ref{prop:algebraic_interlude_1}, we have, for all $f,h\in\mathcal{C}^\infty(M)$ and $x\in M$,
\begin{align*}
e(f,f,h,h)(x) & = \frac{1}{4}(\tes{2f,2h}(x)+4\tes{f,h}(x)-2\tes{2f,h}(x)-2\tes{f,2h}(x))\\
& = \frac{1}{4}(16+4-8-8)\tes{f,h}(x) = \tes{f,h}(x)\geq 0
\end{align*}
This shows that $e$ fulfils Definition~\ref{def:conservative_irreversible}\,(ii) and (iii). Since $\tes{\cdot,\cdot}$ is a left SQPS, we have, for all $f,h\in\mathcal{C}^\infty(M)$,
\begin{align*}
e(h,f,h,h) & = \frac{1}{4}(\tes{f+h,2h} + 2\tes{f,h}+2\tes{h,h}-2\tes{f+h,h}-\tes{f,2h}-\tes{h,2h})\\
& = \frac{1}{4}(\tes{f,2h}+2\tes{f,h}+0-2\tes{f,h}-\tes{2h}-0) = 0
\end{align*}
and hence $e$ fulfils Definition~\ref{def:conservative_irreversible}\,(iv). Finally, Proposition~\ref{prop:algebraic_interlude_1} implies that $e$ is a derivation in each argument. This shows that $E$ is a conservative-irreversible function.
\end{proof}

In particular, we conclude that the distinction between left and right SQPS is artificial: left SQPS are right SQPS and vice versa.

\begin{Corollary}\label{cor:cor_1}
Each left SQPS is symmetric, and therefore a SQPS.
\end{Corollary}
\begin{proof}
Let $\tes{\cdot,\cdot}$ be a SQPS. Then, the function
\begin{align*}
e:\left(\mathcal{C}^\infty(M)\right)^4 & \to\mathcal{C}^\infty(M),\\
(f,g,h,k) & \mapsto \frac{1}{4}\big(\tes{f+g,h+k}+\tes{f,h}+\tes{f,k}+\tes{g,h}+\tes{g,k}\\
& \hspace{1cm}-\tes{f+g,h}-\tes{f+g,k}-\tes{f,h+k}-\tes{g,h+k}\big)
\end{align*}
fulfils the properties (i)--(iv) of Definition~\ref{def:conservative_irreversible}. In view of Lemma~\ref{lem:nice_symmetry}, we conclude, for all $f,g\in\mathcal{C}^\infty(M)$,
\begin{align*}
\tes{f,g} = e(f,f,g,g) = e(g,g,f,f) = \tes{g,f}.
\end{align*}
Thus, $\tes{\cdot,\cdot}$ is symmetric; and hence $\tes{\cdot,h} = \tes{h,\cdot}$ for all $h\in\mathcal{C}^\infty(M)$. This shows that $\tes{\cdot,\cdot}$ is a SQPS.
\end{proof}

In definition~\ref{def:SQPS}, SQPS are not assumed to be symmetric. Since each SQPS is a left SQPS, Corollary~\ref{cor:cor_1} implies that each SQPS is symmetric. An alternative proof, which does not use the one-to-one correspondence of SQPS and conservative-irreversible systems is the following.

\begin{Lemma}\label{lem:symmetric_SQPS}
A function $\tes{\cdot,\cdot}:\mathcal{C}^\infty(M)\times\mathcal{C}^\infty(M)\to\mathcal{C}^\infty(M)$ is a SQPS if, and only if, $\tes{\cdot,\cdot}$ is a symmetric left SQPS.
\end{Lemma}
\begin{proof}
\noindent ``$\impliedby$''\quad By definition, a left SQPS $\tes{\cdot,\cdot}$ is a SQPS if, and only if, the function
\begin{align*}
[[\cdot,\cdot]]:\mathcal{C}^\infty(M)\times\mathcal{C}^\infty(M)\to\mathcal{C}^\infty(M),\quad (f,g)\mapsto \tes{g,f}
\end{align*}
is also a left SQPS. If $\tes{\cdot,\cdot}$ is a symmetric left SQPS, then $[[\cdot,\cdot]] = \tes{\cdot,\cdot}$ and hence the latter is a SQPS.

\noindent ``$\implies$''\quad Conversely, let $\tes{\cdot,\cdot}$ be a SQPS. In particular, $\tes{\cdot,\cdot}$ is a left SQPS and it remains to show symmetry. Let $f,g\in\mathcal{C}^\infty(M)$. Proposition~\ref{prop:algebraic_interlude_1} implies that $\tes{\cdot,0}\equiv 0$ and hence we conclude with property (ii) of Definition~\ref{def:SQPS}
\begin{align*}
\tes{f,g}& = \tes{f+g,g}=\tes{f+g,f+g-f}\\
& = \tes{f+g,f+g}+\tes{f+g,f-f}+\tes{f+g,g-f}-\tes{f+g,f}-\tes{f+g,g}\\
& \quad -\tes{f+g,-f}\\
& = 0+0+\tes{f+g,g-f}-\tes{f+g,f}-\tes{f+g,g}-\tes{f+g,f}\\
& = \tes{2g,g-f}-2\tes{g,f}-\tes{f,g}\\
& = 4\tes{g,g-f}-\tes{f,g}-2\tes{g,f}\\
& = 4\tes{g,f}-\tes{f,g}-2\tes{g,f}=2\tes{g,f}-\tes{f,g},
\end{align*}
from which the symmetry follows.
\end{proof}

As a consequence, we may characterise simple conservative-irreversible functions through their associated SQPS.

\begin{Lemma}
A conservative-irreversible function $E$ is simple if, and only if, the associated SQPS with~\eqref{eq:SQPS_representation} is, up to multiplication with a non-negative function, the square of an almost Poisson bracket.
\end{Lemma}
\begin{proof}
Necessity is evident; we show sufficiency. Let $\lambda:M\to\mathbb{R}$ and let $\{\cdot,\cdot\}$ be an almost Poisson bracket so that
\begin{align*}
\tes{\cdot,\cdot}:\mathcal{C}^\infty(M)\times\mathcal{C}^\infty(M)\to\mathcal{C}^\infty(M),\quad (f,g)\mapsto \lambda\{f,g\}^2
\end{align*}
is an SQPS. Then, the associated conservative-irreversible function $E$ fulfils, for all $f,g,h\in\mathcal{C}^\infty(M)$,
\begin{align*}
E(f,g,h) & = \frac{1}{2}(\tes{f+g,h}-\tes{f,h}-\tes{g,h})\\
& = \lambda\frac{1}{2}(\{f+g,h\}^2-\{f,h\}^2-\{g,h\}^2) = \lambda\{f,h\}\{g,h\},
\end{align*}
and hence $E$ is indeed simple.
\end{proof}

\section{Dynamical systems of conservative-irreversibly functions}\label{sec:5}

Let $M$ be a smooth manifold and $E$ a conservative-irreversible function on $M$. Each pair of generating functions $(s,h)\in\mathcal{C}^\infty(M)\times\mathcal{C}^\infty(M)$ induces the dynamical system
\begin{align}\label{eq:cons_irrev_system}
\tfrac{\mathrm{d}}{\mathrm{d}t} f\circ x = E(f,s,h)\circ x\qquad \text{for all}~f\in\mathcal{C}^\infty(M),
\end{align}
i.e. a smooth curve $x\in\mathcal{C}^\infty(M)$ is a trajectory of the dynamical system~\eqref{eq:cons_irrev_system} if, and only if, for each \textit{observable} $f\in\mathcal{C}^\infty(M)$ the curve $f\circ x$ fulfils the equation~\eqref{eq:cons_irrev_system}. Using the associated tensor field $\varepsilon\in\Gamma(T^4(M))$, we may coin the following definition.

\begin{Definition}
Let $E$ be a conservative-irreversible function on a manifold $M$. A vector field $X$ is called \textit{conservative-irreversible} with generating functions $(s,h)\in\mathcal{C}^\infty(M)\times\mathcal{C}^\infty(M)$ if, and only if, $X = E(\cdot,s,h)$; we write $X_{(s,h)} := E(\cdot,s,h)$.
\end{Definition}

Analogously to Hamiltonian systems defined on Poisson manifolds, the dynamical system~\eqref{eq:cons_irrev_system}, written implicitly via the infinitesimal action of the solutions on all observables, can then be recognised as the dynamical system given by the vector field $X_{(s,h)}$.

\begin{Lemma}
Let $s,h\in\mathcal{C}^\infty$. A curve $x\in\mathcal{C}^\infty(I,M)$ is a trajectory of~\eqref{eq:cons_irrev_system} if, and only if, $x$ is an integral curve of $X_{(s,h)}$.
\end{Lemma}
\begin{proof}
``$\impliedby$''\quad Let $x$ be an integral curve of $X_{(s,h)}$ and $f\in\mathcal{C}^\infty(M)$. Then, we have
\begin{align*}
\tfrac{\mathrm{d}}{\mathrm{d}t} f\circ x = \left\langle\mathrm{d}f\circ x,\tfrac{\mathrm{d}}{\mathrm{d}t} x\right\rangle = \left\langle \mathrm{d}f\circ x,X\circ x\right\rangle = \varepsilon(\mathrm{d}f,\mathrm{d}s,\mathrm{d}h,\mathrm{d}h)\circ x = E(f,s,h)\circ x
\end{align*}
and thus $x$ is a trajectory of the dynamical system~\eqref{eq:cons_irrev_system}.

``$\implies$''\quad Conversely, let $x\in\mathcal{C}^\infty(I,M)$, $I\subseteq\mathbb{R}$ a nonempty open interval, be a trajectory of the dynamical system~\eqref{eq:cons_irrev_system} and let $t_0\in I$. Let $\varphi_1,\ldots,\varphi_n\in\mathcal{C}^\infty(M)$ so that $(\varphi_1,\ldots,\varphi_n)$ is a diffeomorphism in a neighbourhood of $x(t_0)$. Then, we have in a neighbourhood $J\subseteq I$ of $t_0$
\begin{align*}
\forall i\in\mathbb{N}^*_{\leq n}: \tfrac{\mathrm{d}}{\mathrm{d}t}\varphi_i\circ x = E(\varphi_i,s,h)\circ x = (X_{(s,h)})_i\circ x,
\end{align*}
where $(X_{(s,h)})_i$ is the $i$-th coordinate with respect to the local coordinate system $(\partial_{\varphi_1},\ldots,\partial_{\varphi_n})$. This shows that $x$ is indeed an integral curve of $X_{(s,h)}$.
\end{proof}

In the following lemma, we verify that the dynamical system~\eqref{eq:cons_irrev_system} fulfils the laws of thermodynamics. Additionally, we show that the entropy function $s$ is, in general, not preserved.

\begin{Lemma}
Let $E$ be a conservative-irreversible function on $M$, $s,h\in\mathcal{C}^\infty(M)$ and let $x$ be a solution of~\eqref{eq:cons_irrev_system}. Then, we have
\begin{enumerate}[(i)]
\item $\tfrac{\mathrm{d}}{\mathrm{d}t} h\circ x\equiv 0$, i.e. (generalised) energy is preserved along trajectories.
\item $\tfrac{\mathrm{d}}{\mathrm{d}t} s\circ x\geq 0$, i.e. (generalised) entropy is irreversibly produced along trajectories.
\end{enumerate}
Furthermore, we see
\begin{enumerate}
\item[(iii)] If $E\neq 0$, then there exists generating function $s_0,h_0\in\mathcal{C}^\infty(M)$ so that $\tfrac{\mathrm{d}}{\mathrm{d}t} s_0\circ x\neq 0$ for some solution $x$ of~\eqref{eq:cons_irrev_system}.
\end{enumerate}
\end{Lemma}
\begin{proof}
Let $E$ be a conservative-irreversible function and $s,h\in\mathcal{C}^\infty(M)$. The properties (i) and (ii) are direct consequences of Definition~\ref{def:conservative_irreversible}\,(iii) and (iv):
\begin{align*}
\tfrac{\mathrm{d}}{\mathrm{d}t} h\circ x = E(h,s,h,h)\circ x \equiv 0
\end{align*}
and
\begin{align*}
\tfrac{\mathrm{d}}{\mathrm{d}t} s\circ x = E(s,s,h)\circ x \geq 0.
\end{align*}
Now, let $E\neq 0$. Due to symmetry of $E$ in the first two arguments, there exists $s_0,h_0\in \mathcal{C}^\infty(M)$ and $x_0\in M$ so that $E(s_0,s_0,h_0)(x_0)\neq 0$. Consider an integral curve $x$ of the vector field $X_{(s_0,h_0)}$ with $x(0) = x_0$. Then, we have 
\begin{align*}
\tfrac{\mathrm{d}}{\mathrm{d}t}(s\circ x)(0) = E(s_0,s_0,h_0)(x_0)\neq 0.
\end{align*}
This completes the proof.
\end{proof}

Using the conservative irreversible functions, an irreversible Hamiltonian system may be defined in an alternative and slightly more general way to~\cite{RamiMascSbar13,RamiLeGo22} as follows.

\begin{Definition}
Let $M$ be a smooth manifold with Poisson bracket $\{\cdot,\cdot\}$ and conservative-irreversible function $E$. An \textit{irreversible Hamiltonian system} with generating functions $(s,h)\in\mathcal{C}^\infty(M)\times\mathcal{C}^\infty(M)$ is the dynamical system
\begin{align}\label{eq:irr_Ham}
\tfrac{\mathrm{d}}{\mathrm{d}t}f\circ x = \{f,h\} + E(f,s,h)\qquad \text{for all}~f\in\mathcal{C}^\infty(M),
\end{align}
with either the \textit{weak noninteraction condition} $\{s,h\} = 0$ or the \textit{strong noninteraction condition} $\{s,\cdot\}\equiv 0$. The vector field $\widehat{X}_{(s,h)} := X_h+X_{(s,h)}$, where~$X_h$ denotes the Hamiltonian vector field $X_h = \{\cdot,h\}$ and $X_{(s,h)}$ is the conservative-irreversible vector field generated by $(s,h)$, is called \textit{irreversible Hamiltonian vector field}.
\end{Definition}

Lastly, we demonstrate how irreversible Hamiltonian systems may be viewed as particular classes of metriplectic systems as well as the general systems of Beris and Edwards, thereby unifying, in some sense, these approaches.

Let $M$ be a smooth manifold with Poisson bracket $\{\cdot,\cdot\}$ and conservative-irreversible function $E$. We show how the irreversible Hamiltonian vector fields can be viewed simultaneously as metriplectic systems and as systems of Beris and Edwards. Let $s,h\in\mathcal{C}^\infty(M)$ so that the noninteraction condition $\{s,h\} = 0$ is fulfilled. Define
\begin{align*}
[\cdot,\cdot]_{BE}:=E(\cdot,s,\cdot)
\end{align*}
Then, $[\cdot,\cdot]_{BE}$ is a derivation in the first argument and fulfils, by construction,
\begin{align*}
[h,h]_{BE} = 0 \leq [s,h]_{BE}
\end{align*}
so that the system~\eqref{eq:irr_Ham} written as
\begin{align*}
\tfrac{\mathrm{d}}{\mathrm{d}t}f\circ x = \{f,h\} + [s,h]_{BE}\qquad \text{for all}~f\in\mathcal{C}^\infty(M),
\end{align*}
has the form which was considered by~\cite{BeriEdwa94}. Secondly, consider the brackets
\begin{align*}
[\cdot,\cdot]_M:=E(\cdot,\cdot,h).
\end{align*}
Definition~\ref{def:conservative_irreversible}\,(iii) guarantees that this is a non-negative bracket and
\begin{align*}
[h,\cdot]_M = E(h,\cdot,h) \equiv 0
\end{align*}
shows that $[\cdot,\cdot]_M$ fulfils the strong noninteraction condition of metriplectic systems. Since the (weak or strong) noninteraction condition on the almost Poisson bracket are the same as for irreversible Hamiltonian systems,~\eqref{eq:irr_Ham} written as
\begin{align*}
\tfrac{\mathrm{d}}{\mathrm{d}t}f\circ x = \{f,h\} + [f,s]_M\qquad \text{for all}~f\in\mathcal{C}^\infty(M),
\end{align*}
is a metriplectic system. 

Insofar, irreversible Hamiltonian systems are indeed simultaneously particular  metriplectic systems and Beris-Edwards-systems. However, this view is a bit flawed. The brackets $[\cdot,\cdot]_{BE}$ and $[\cdot,\cdot]_M$ depend explicitly on $s$ and $h$, respectively. Therefore, a perturbation of the generating functions must automatically be a perturbation of the underlying brackets in this point of view, which is not the case for general systems.

\section{Conclusion}\label{sec:6}

We have studied conservative-irreversible functions on smooth, finite-dimensional manifolds. These functions are uniquely associated to a class of four-tensor-fields and a class of bi-quadratic functions, which were fully characterised by intrinsic properties. As an interesting subclass, we recognised simple conservative-irreversible systems, which are induced by (almost) Poisson brackets and which were the original motivation of the study of these systems. Using local representations, we were able to show that each conservative-irreversible function is a locally finite, but in general not trivial, sum of simple conservative-irreversible functions, which revealed the relation to the recently suggested class of metriplectic four-tensors.

\bibliographystyle{plain}   
\bibliography{Literatur}

\begin{thebibliography}{10}

\bibitem{Arno89}
V.~I. Arnold.
\newblock {\em Mathematical Methods of Classical Mechanics}.
\newblock Springer, second edition, 1989.

\bibitem{BeriEdwa94}
Anthony~N. Beris and Brian~J. Edwards.
\newblock {\em Thermodynamics of flowing systems with internal microstructure}.
\newblock Oxford University Press, 1994.

\bibitem{Butt07}
J.~Butterfield.
\newblock On symplectic reduction in classical mechanics.
\newblock In Jeremy Butterfield and John Earman, editors, {\em Philosophy of
  Physics}, Handbook of the Philosophy of Science, pages 1--131. North-Holland,
  Amsterdam, 2007.

\bibitem{CantLeonDieg99}
F.~Cantrijn, M.~{de León}, and D.~{Martín de Diego}.
\newblock On almost-{P}oisson structures in nonholonomic mechanics.
\newblock {\em Nonlinearity}, 12:721--737, 1999.

\bibitem{EdwaBeri91}
Brian~J. Edwards and Antony~N. Beris.
\newblock Unified view of transport phenomena based on the generalized bracket
  formulation.
\newblock {\em Industrial \& Engineering Chemistry Research}, 30(5):873--881,
  1991.

\bibitem{ElmaKarpMerk08}
Richard Elman, Nikita Karpenko, and Alexander Merkurjev.
\newblock {\em The Algebraic and Geometric Theory of Quadratic Forms}.
\newblock American Mathematical Society, 2008.

\bibitem{GreuHalpVans72}
Werner Greub, Stephen Halperin, and Ray Vanstone.
\newblock {\em Connections, Curvature, and Cohomology}.
\newblock Academic Press, New York and London, 1972.

\bibitem{Grme84}
Miroslav Grmela.
\newblock Bracket formulation of dissipative fluid mechanics equations.
\newblock {\em Physics Letters A}, 102(8):355--358, 1984.

\bibitem{Guha07}
Partha Guha.
\newblock Metriplectic structure, {L}eibniz dynamics and dissipative systems.
\newblock {\em Journal of Mathematical Analysis and Applications}, 326:121 --
  136, 2007.

\bibitem{HelmMicaRevo14}
Jacques Helmstetter, Artibano Micali, and Philippe Revoy.
\newblock On the {D}efinition of {Q}uadratic {M}appings.
\newblock {\em Advances in Applied Clifford Algebras}, 24:125 -- 140, 2014.

\bibitem{Kauf84}
Allan~N. Kaufman.
\newblock Dissipative {H}amiltonian systems: A unifying principle.
\newblock {\em Physics Letters A}, 100(8):419--422, 1984.

\bibitem{MascKirc23b}
Jonas Kirchhoff and Bernhard Maschke.
\newblock On the generating functions of irreversible port-{H}amiltonian
  systems.
\newblock {\em IFAC-PapersOnLine}, 56:10447--10452, 2023.

\bibitem{Kneb10}
Manfred Knebusch.
\newblock {\em Specialsation of Quadratic and Symmetric Bilinear Forms}.
\newblock Springer, 2010.

\bibitem{LibeMarl87}
Paulette Libermann and Charles-Michel Marle.
\newblock {\em Symplectic Geometry and Analytical Mechanics}.
\newblock Mathematics and Its Applications. D. Reidel Publishing Company,
  Dordrecht, Boston, Lancaster, Tokyo, 1987.

\bibitem{LiebYngva98}
Elliot~H. Lieb and Jakob Yngvason.
\newblock A guide to {E}ntropy and the {S}econd {L}aw of {T}hermodynamics.
\newblock {\em Notices of the AMS}, 45(5):571--581, 1998.

\bibitem{MarsRatiu99}
Jerrold~E.\ Marsden and Tudor~S.\ Ratiu.
\newblock {\em Introduction to Mechanics and Symmetry}.
\newblock Springer Verlag, New York, Berlin, Heidelberg, 2nd edition, 1999.

\bibitem{Morr84}
Philip~J. Morrison.
\newblock Bracket formulation for irreversible classical fields.
\newblock {\em Physics Letters A}, 100(8):423--427, 1984.

\bibitem{Morr86}
Philip~J. Morrison.
\newblock A paradigm for joined {H}amiltonian and dissipative systems.
\newblock {\em Physica D: Nonlinear Phenomena}, 18(1):410--419, 1986.

\bibitem{Morr09}
Philip~J. Morrison.
\newblock Thoughts on brackets and dissipation: old and new.
\newblock {\em Journal of Physics: Conference Series}, 169:012006, 2009.

\bibitem{MorrUpdi23}
Philip~J. Morrison and Michael~H. Updike.
\newblock An inclusive curvature-like framework for describing dissipation.
\newblock {\em arxiv-version~\url{https://arxiv.org/pdf/2306.06787.pdf}}, 2023.

\bibitem{OrtePlBi04}
Juan-Pablo Ortega and Victor Planas-Bielsa.
\newblock Dynamics on {L}eibniz manifolds.
\newblock {\em Journal of Geometry and Physics}, 52:1--27, 2004.

\bibitem{RamiLeGo22}
Hector Ramirez and Yann Le~Gorrec.
\newblock An {O}verview on {I}rreversible {P}ort-{H}amiltonian {S}ystems.
\newblock {\em Entropy}, 24(10), 2022.

\bibitem{RamiMascSbar13}
Hector Ramirez, Bernhard Maschke, and Daniel Sbarbaro.
\newblock Irreversible port-{H}amiltonian systems: A general formulation of
  irreversible processes with application to the cstr.
\newblock {\em Chemical Engineering Science}, 89:223--234, 2013.

\bibitem{Oett05}
Hans~Christian Öttinger.
\newblock {\em Beyond Equilibrium Thermodynamics}.
\newblock Wiley-Interscience, 2005.

\end{thebibliography}

\end{document}